\documentclass[10pt,journal,twocolumn,twoside]{IEEEtran}  
\IEEEoverridecommandlockouts                              

\usepackage{amsmath} 
\usepackage{amssymb}  
\usepackage{amsthm}

\usepackage{color}
\usepackage{graphicx}
\usepackage{comment}
\usepackage{cite}
\usepackage{subcaption}
\usepackage{url}
\usepackage{hyperref}
\usepackage{enumitem}

\newcommand{\dP}{\mathrm{d}\mathbb{P}}
\newcommand{\dQ}{\mathrm{d}\mathbb{Q}}
\newcommand{\dG}{\mathrm{d}\mathbb{G}}

\newcommand{\T}{\top}
\newcommand{\I}{\mathcal{I}}
\newcommand{\KL}{\mathrm{KL}}
\newcommand{\Cov}{\text{Cov}}
\newcommand{\Tr}{\textnormal{Tr}}
\newcommand{\etal}{et al.}
\newcommand{\SI}{\text{SI}}

\newcommand{\rx}{\bold{x}}
\newcommand{\ry}{\bold{y}}
\newcommand{\ru}{\bold{u}}
\newcommand{\ra}{\bold{a}}
\newcommand{\rw}{\bold{w}}
\newcommand{\rc}{\bold{c}}
\newcommand{\rh}{\bold{h}}
\newcommand{\rv}{\bold{v}}
\newcommand{\rz}{\bold{z}}

\newcommand{\rn}{\bold{n}}

\newtheorem{remark}{Remark}
\newtheorem{thm}{Theorem}

\newtheorem{defn}{Definition}
\newtheorem{lemma}{Lemma}

\newtheorem{prop}{Proposition}
\newtheorem{corollary}{Corollary}

\title{\LARGE \bf
Minimum Rate for Partially Observable Linear Systems with Side Information: LQG Plants and Gaussian-Markov Sources
}

\author{Sijie Li and Hyeji Kim
\thanks{A preliminary version of the paper will be presented at the CDC 2026.}%
\thanks{Sijie Li and Hyeji Kim are with the Department of Electrical and Computer Engineering,
        University of Texas at Austin
        {\tt\small \{sijieli,hyeji\}@utexas.edu}}%
}

\begin{document}

\maketitle
\thispagestyle{empty}
\pagestyle{empty}

\begin{abstract}
This paper studies the minimum rate required for a partially observable linear system with side information. The linear quadratic Gaussian (LQG) plant and the Gaussian-Markov source are considered. Assume the encoder uses prefix codes. We show that a class of linear policies is sufficient for optimizing the causally conditioned directed information rate lower bound. We also show that the resulting optimization problems admit convex formulations for both finite-horizon time-varying and infinite-horizon time-invariant scalar systems. Our results generalize previous results for systems with full or partial observation only, and systems with full observation and side information. Numerical simulations are presented to illustrate the effect of side information for partially observable systems.
\end{abstract}

\section{INTRODUCTION}

In this paper, we study the minimum amount of information needed for a linear system, i.e., a Linear Quadratic Gaussian (LQG) plant or a Gaussian-Markov source, to achieve a certain distortion requirement with partial observation and side information. At each time step, the sensor/encoder receives a noisy measurement of the plant state and sends a codeword to the controller/decoder through a noiseless communication link using prefix codes. The controller/decoder receives the codeword and the additional side information and outputs a control signal or an estimate of the plant state. Both encoder and decoder operate with their local history information only.
This problem is both interesting and challenging, as achieving the optimal tradeoff necessitates a co-design of communication and control. Moreover, our problem setup generalizes several prior models, including settings with full or partial observation only, as well as those with full observation and side information.

 The sequential rate-distortion function (SRDF), which serves as a lower bound on the operational rate, has been extensively studied for the {single} encoder and {single} controller/decoder case \cite{GP1973,GP1974,lqg-scalar,DO2012,1.254-gap,vector-achievable-scheme,GMsdp,lqgsdp,lqg-epi,LK2020,CCKS2022,LQG-sideinfo1,LQG-sideinfo2,partial-observation-GM,longversion}. Nonetheless, prior work focuses on scenarios with either partial observation or side information (or neither), but not both. The study of the SRDF dates back to Gorbunov and Pinsker \cite{GP1973,GP1974}. They introduce the notion of nonanticipatory epsilon entropy and derive analytical expressions for Gaussian-Markov sources. The complete characterization for multivariate Gaussian-Markov sources is given by Charalambous \etal \cite{CCKS2022}. Tatikonda \etal \cite{lqg-scalar} introduce the SRDF framework to control systems. Derpich and {\O}stergaard \cite{DO2012} construct achievable schemes for Gaussian-Markov sources with rates at most 1.254 bits away from the lower bound. This is generalized to the LQG plant by Silva \etal \cite{1.254-gap} and to vector cases by Tanaka \etal \cite{vector-achievable-scheme}. Tanaka \etal \cite{GMsdp} apply semidefinite programming (SDP) approaches for computing the SRDF for the Gaussian-Markov sources. This is later generalized to the LQG plant by Tanaka \etal \cite{lqgsdp}. In \cite{lqgsdp}, they also extend the approach to partially observable systems. Stavrou and Skoglund \cite{partial-observation-GM} study the partially observable Gaussian-Markov sources. Li \etal \cite{longversion} extend the SDP formulations to a broader class of parameters. Another work by Kostina and Hassibi \cite{lqg-epi} uses the conditional entropy power inequality to establish an analytical lower bound for the SRDF. For the case of full observation with side information, Cuvelier and Tanaka \cite{LQG-sideinfo2} study the case when part of the full observation is sent to the controller as side information. Sabag \etal \cite{LQG-sideinfo1} study the case with additional noisy side information to the controller. They both construct SDP formulations for the SRDF in these cases. Lev and Khina \cite{LK2020} study the scalar Gaussian-Markov source with side information.

A few results are available for the \textit{multiple} sensor and \textit{single} controller/decoder setting. Jung \etal \cite{lqgceo} consider the case with multiple sensors with partial observations but with \textit{feedback} from the controller/decoder to the encoders. Li \etal \cite{multiple-sensor-lower-bound} study the case without the \textit{feedback} and show a causally conditioned directed information lower bound and the sufficiency of a class of linear policies for optimizing the weighted sum lower bound.

In this paper, we study the case when the encoder has {\em partial observation} and the controller/decoder has {\em side information}. This information structure introduces additional challenges, as the encoder cannot recover the global state estimate with the partial observations only. Consequently, the reduction to a fully observable system as in \cite{lqgsdp} and \cite{partial-observation-GM} is not directly applicable here. To tackle this challenge, we note that the control in the LQG plant and the estimation in the Gaussian-Markov source can be formulated in terms of an {\em observable Markov chain}, rather than a hidden one. The state of the {\em observable Markov chain} is the global state estimate of the hidden state, i.e., the conditional mean of the hidden state given the partial observation, the side information, and the past control signals, and the noises that drive the {\em observable Markov chain} are also Gaussian and independent. With this observation, we show that it is sufficient to consider a class of linear policies to minimize the causally conditioned directed information lower bound on the rate from \cite{multiple-sensor-lower-bound} with respect to a certain distortion requirement.

This result generalizes the sufficient linear policy structure from the past works \cite{GMsdp,lqgsdp,CCKS2022,LQG-sideinfo1,LQG-sideinfo2,partial-observation-GM}. We further show that the resulting optimization problems admit convex formulations in the scalar case, for both time-varying and time-invariant systems. For the infinite-horizon time-invariant systems, we have a single-letter convex optimization form. A new technical challenge here is to show the convergence of two coupled error covariance recursions. 

The main contributions of the paper can be summarized as follows:
\begin{itemize}
    \item 
    We show that it is sufficient to consider a class of linear policies for optimizing the causally conditioned directed information lower bound for systems with partial observations at the encoder and additional side information at the controller/decoder in the general vector case.
    The sufficient linear policy structure recovers the corresponding structural results of past works \cite{GMsdp,lqgsdp,CCKS2022,LQG-sideinfo1,LQG-sideinfo2,partial-observation-GM}, including the cases with full observation with or without side information and the cases with partial observation only. See Theorem \ref{thm:linear-policy-opt} for the LQG plant and Corollary \ref{cor:optimal-policy-GM} for the Gaussian-Markov source.
    \item We further study the optimization problem within the sufficient linear policy set for scalar systems. 
    We show that the optimization problem admits a convex formulation for the time-varying systems over a finite horizon. See Theorem \ref{thm:convex-time-varying} for the LQG plant and Section \ref{sec:G-M-results} for the Gaussian-Markov source.
    \item For the infinite-horizon time-invariant scalar systems, we establish an equivalent single-letter convex formulation. See Theorem \ref{thm:convex-TI} for the LQG plant and Section \ref{sec:G-M-results} for the Gaussian-Markov source. 
\end{itemize}

\textbf{Notations.} Lowercase letters denote vectors and uppercase letters denote matrices. Boldface lowercase letters denote random variables. We denote the transpose of $x$ as $x^{\T}$ and the Moore-Penrose pseudoinverse of $X$ as $X^{\dagger}$. The subscripts represent time steps. The superscripts denote histories up to the indicated time. For example, $x_t$ is a realization at time step $t$ and $x^{t} = [x_{1}^{\T} ,\ldots ,x_{t}^{\T}]^{\T}$.  Let $(\rx,\ry)$ be two random variables on $(\mathcal{X}\times \mathcal{Y},\mathcal{B}_{\mathcal{X}}\otimes\mathcal{B}_{\mathcal{Y}})$ and let $\mathbb{P}_{\rx,\ry}$ and $\mathbb{G}_{\rx,\ry}$ be two joint measures. If $\mathbb{P}_{\rx|y}\ll\mathbb{G}_{\rx|y}$ holds $\mathbb{P}_{\ry}$ a.s., the conditional KL divergence is defined as 
\begin{align*}
    \KL_{\mathbb{P},\mathbb{G}} (\rx|\ry) =
    \int \log\left(\frac{\dP(x|y)}{\dG(x|y)}\right)\dP(x|y)\dP(y)
\end{align*}
and $\infty$ otherwise. Consider three random variables $(\rx,\ry,\rz)$ with joint measure $\mathbb{P}$. If $\mathbb{P}_{\rx,\ry|z}\ll \mathbb{P}_{\rx|z}\otimes\mathbb{P}_{\ry|z}$ holds $\mathbb{P}_{\rz}$ a.s., the conditional mutual information is defined as
\begin{align*}
    I(\rx;\ry|\rz) = \int \log\left(\frac{\dP(x,y|z)}{\dP(x|z )\otimes \dP(y|z)}\right)\dP(x,y,z)
\end{align*}
and $\infty$ otherwise. Define the causally conditioned directed information \cite{Kramer2003} as
\begin{align*}
    I(\rx^T \rightarrow \ry^T \Vert \rz^T) = \sum_{t=1}^T I(\rx^t;\ry_t | \ry^{t-1} , \rz^t)
\end{align*}


\section{Problem Model}

\begin{figure}[htbp]
    \centering
    \includegraphics[width=1\linewidth]{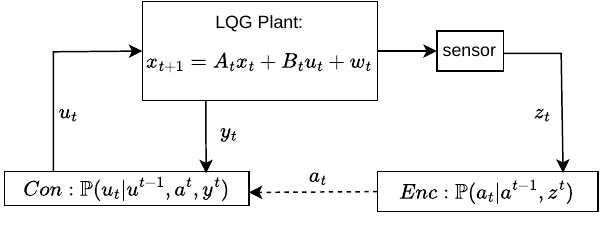}
    \caption{Problem model of the LQG plant.}
    \label{fig:LQG-model}
\end{figure}
In this section, we present the problem models for LQG plants and Gaussian-Markov sources in Sections \ref{sec:lqg-model} and \ref{sec:GM-model}, respectively. The corresponding models are illustrated in Figures \ref{fig:LQG-model} and \ref{fig:GM-model}.

\subsection{Linear Quadratic Gaussian Plant}\label{sec:lqg-model}
Consider an LQG plant with the evolving equation:
\begin{equation}\label{eqn:hidden-process}
    \rx_{t+1} = A_{t}\rx_{t} + B_{t}\ru_{t}+\rw_{t}
\end{equation}
where $x_{t}\in\mathbb{R}^{n}$ is the plant state, $u_{t}\in\mathbb{R}^{m}$ is the control input, $A_{t}\in\mathbb{R}^{n\times n}$, $B_{t}\in\mathbb{R}^{n \times m}$ are given matrices and $\rw_t \sim \mathcal{N}(0,W_{t})$ is an independent Gaussian noise with $W_t\succeq 0$. We assume the initial distribution $\rx_1\sim\mathcal{N}(0,P_{1|0})$ with $P_{1|0}\succ 0$. Consider two partial observations $\rz_{t}$ and $\ry_{t}$:
\begin{align}
    \rz_{t} &= C_{t}\rx_{t} + \rn_{t}\label{eqn:enc-observation}\\
    \ry_t &= F_t \rx_t + \rv_t\label{eqn:side-info}
\end{align}
where $C_t \in \mathbb{R}^{l_1\times n}$, $F_t\in\mathbb{R}^{l_2\times n}$, $\rn_t\sim \mathcal{N}(0,N_t)$ and $\rv_t\sim\mathcal{N}(0,V_t)$ are independent Gaussian noises with $N_t\succeq 0$ and $V_t\succeq 0$. In this paper, we consider the scenario in Figure \ref{fig:LQG-model}. Let $T$ be the total number of time steps. The encoder and the decoder/controller operate as follows:
\begin{enumerate}[label=(\alph*),ref=(\alph*)]
    \item \textbf{Encoder}: encoder observes $z_t$, then generates a codeword $a_t\in\mathcal{A}_t$ with a prefix code \cite{cover2012elements} based on the history information $(a^{t-1},z^{t})$, where $\mathcal{A}_t\subset \{0,1,00,01,\dots\}$ is the alphabet for $a_t$. \label{def:encoder}
    \item \textbf{Decoder/Controller}: decoder/controller receives $a_t$ and side information $y_t$, then generates a control signal $u_t$ based on $(u^{t-1},a^{t},y^{t})$. \label{def:controller}
\end{enumerate}
The feasible policy set $\Omega'$ is characterized by the following stochastic kernels:
\begin{align}
    \mathbb{P}(a^{T}\Vert z^{T}) &\triangleq \prod_{t=1}^{T}\mathbb{P}(a_t|a^{t-1},z^{t})\label{encoder}\\
    \mathbb{P}(u^{T}\Vert a^{T},y^{T}) &\triangleq \prod_{t=1}^{T}\mathbb{P}(u_t|u^{t-1},a^{t},y^{t}) \label{controller}
\end{align}
In particular, the feasible policy set $\Omega'$ implies the following Markov chain condition
\begin{align}\label{markov-chain-condition}
    \ru_t - (\ru^{t-1},\rz^t,\ry^t)-\rx_t
\end{align}
Define the rate and the control cost as follows:
\begin{align*}
    R &= \frac{1}{T}\sum_{t=1}^T \mathbb{E}[l(\ra_t)] \\
    J(\rx^{T+1},\ru^T) &= \frac{1}{T} \sum_{t=1}^{T}\mathbb{E}[\Vert \rx_{t+1}\Vert_{Q_t}^{2} + \Vert \ru_t \Vert_{R_t}^{2}] 
\end{align*}
where $l(\cdot)$ measures the length of the codeword $\ra_t $ in terms of bits, $\Vert x_{t+1}\Vert_{Q_t}^{2} = x_{t+1}^{\T}Q_t x_{t+1}$, $ Q_t\succ 0$ and $R_t\succ 0 $.
From \cite{multiple-sensor-lower-bound}, we can show the following lower bound on the rate.

\begin{lemma}\label{lemma:DIlowerbound} Suppose the encoder uses prefix codes. The following causally conditioned directed information serves as a lower bound for the rate \cite{multiple-sensor-lower-bound}.
\begin{align*}
    R\geq \frac{1}{T}I(\rz^{T}\rightarrow \ru^T \Vert \ry^T)
\end{align*}
\end{lemma}

Let $d^{*}$ be the theoretically minimal control cost and consider $d>d^{*}$.
Let $\Omega\triangleq \{\mathbb{P}(u_t|u^{t-1},z^{t},y^{t})\}_{t=1}^{\infty}$ be the set of feasible policies for the lower bound in Lemma \ref{lemma:DIlowerbound}. We consider the following optimization problem for finite $T$.
\begin{align}\label{LQG-plain-opt-form}
    \inf_{\omega\in\Omega, J\leq d} I(\rz^{T}\rightarrow \ru^T \Vert \ry^T)
\end{align}

\subsection{Gaussian-Markov Source}\label{sec:GM-model}

In this setting, we consider the Gaussian-Markov process, which can be viewed as an uncontrolled version of \eqref{eqn:hidden-process}:
\begin{align}\label{eqn:hidden-process-GM}
    \rx_{t+1} = A_t \rx_t +  \rw_t
\end{align}

\begin{figure}[htbp]
    \centering
    \includegraphics[width=1\linewidth]{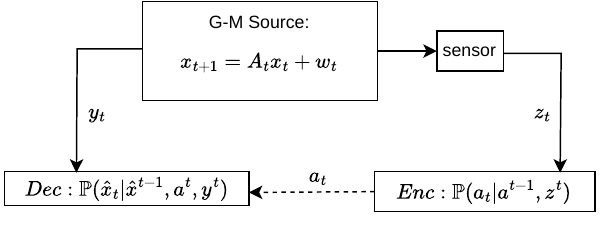}
    \caption{Problem model of the Gaussian-Markov source.}
    \label{fig:GM-model}
\end{figure}

We also have partial observations $\rz_t$ for the encoder and $\ry_t$ for the decoder as side information. $\rz_t$ follows \eqref{eqn:enc-observation} and $\ry_t$ follows \eqref{eqn:side-info}. The problem model is as in Figure \ref{fig:GM-model}. We also assume the use of prefix codes \cite{cover2012elements}. The encoder operates the same as the encoder~\ref{def:encoder}. The decoder operates similarly to the decoder/controller~\ref{def:controller} except that the output signal is the estimation $\hat{x}_t$.
The feasible policy set $\Lambda'$ is 
\begin{align}
    \mathbb{P}(a^{T}\Vert z^{T}) &\triangleq\prod_{t=1}^{T} \mathbb{P}(a_t|a^{t-1},z^{t})\label{encoder-GM}\\
    \mathbb{P}(\hat{x}^{T}\Vert a^{T},y^{T}) &\triangleq \prod_{t=1}^{T}\mathbb{P}(\hat{x}_t|\hat{x}^{t-1},a^{t},y^{t}) \label{decoder-GM}
\end{align}
 We are interested in the tradeoff between the rate, which is defined the same as in the LQG setting, and the weighted mean-square error (WMSE) distortion:
\begin{align*}
    \text{WMSE}(\rx^T , \hat{\rx}^T) = \frac{1}{T}\sum_{t=1}^{T}\mathbb{E}[(\rx_t - \hat{\rx}_t)^\T \Xi_t (\rx_t - \hat{\rx}_t)]
\end{align*}
where $\Xi_t\succ 0$ is the weight matrix for mean-square error. 

The argument for Lemma \ref{lemma:DIlowerbound} applies to the Gaussian-Markov source case by replacing $u_t$ with $\hat{x}_t$. Hence, the causally conditioned directed information lower bound also holds. And we are interested in the following optimization problem for finite $T$ where $\Lambda\triangleq \{\mathbb{P}(\hat{x}_t|\hat{x}^{t-1},z^{t},y^{t})\}_{t=1}^{\infty}$.
\begin{align}\label{GM-plain-opt-form}
    \inf_{\lambda\in\Lambda, \text{WMSE}\leq d} I(\rz^{T}\rightarrow \hat{\rx}^T \Vert \ry^T)
\end{align}

\section{Observable Markov Chain}

For this partially observable system, the key step is to identify the \textit{observable Markov chain}. The same idea is used in \cite{lqgsdp} for LQG plants and in \cite{partial-observation-GM} for Gaussian-Markov sources, both with partial observations only. In our setting, this idea is applied to the global state estimate based on both partial observations and side information. 
 We identify and establish some properties of the \textit{observable Markov chain} for the LQG case. The Gaussian-Markov source case can be viewed as the LQG case with $u_t = 0$ for all $t$.
\subsection{Linear Quadratic Gaussian Plant}

We first define the information set at time step $t $ as 
\begin{align*}
    \I_t = \{\ry^t,\rz^t,\ru^t\}.
\end{align*}
The realization of $\I_t$ is denoted as $i_t = \{y^{t},z^{t},u^{t}\}$. Note that the information set $\I_t$ contains all the information that the encoder and the controller have jointly about the linear process at each time step. Denote the conditional expectation of $\rx_t$  and $\rx_{t+1}$ given $\I_t$ as:
\begin{align*}
    \Bar{\rx}_t &\triangleq \mathbb{E}[\rx_t|\I_t]\\
    \Bar{\rx}_{t+1|t} &\triangleq \mathbb{E}[\rx_{t+1}|\I_t]
\end{align*}
Since the Markov chain condition \eqref{markov-chain-condition} holds, we also have
\begin{align*}
    \bar{\rx}_t = \mathbb{E}[\rx_t|\ry^t,\rz^t,\ru^{t-1}]
\end{align*}
Denote 
 $\Bar{\ry}_t = \begin{bmatrix}
    \ry_t \\
    \rz_t
\end{bmatrix}$ with $\Bar{F}_t = \begin{bmatrix}
    F_t \\
    C_t
\end{bmatrix}$, $\Bar{\rv}_t = \begin{bmatrix}
    \rv_t \\
    \rn_t
\end{bmatrix}$, and the covariance matrix of $\Bar{\rv}_t$ as $\Bar{V}_t = \text{diag}(V_t,N_t)$. The error covariance matrices of $\rx_t$ and $\rx_{t+1}$ given $\I_t$  are denoted as follows:
\begin{align*}
    \Bar{W}_t &= \mathbb{E}[(\rx_t - \Bar{\rx}_t)(\rx_t - \Bar{\rx}_t)^{\T}]\\
    \Bar{W}_{t+1|t} &= \mathbb{E}[(\rx_{t+1} - \Bar{\rx}_{t+1|t})(\rx_{t+1} - \Bar{\rx}_{t+1|t})^{\T}]
\end{align*}
Following the plant state evolution equation \eqref{eqn:hidden-process}, we have 
\begin{align}
    \Bar{\rx}_{t+1|t} = A_t \Bar{\rx}_t + B_t \ru_t
\end{align} which immediately implies that 
\begin{align}
    \Bar{W}_{t+1|t} =A_t\Bar{W}_t A_t^{\T} + W_t
\end{align}
The following update equation for $\bar{\rx}_t$ is true due to Proposition \ref{prop:decomposition}:
\begin{align}\label{eqn:kalman-update}
    \Bar{\rx}_t = \Bar{\rx}_{t|t-1} + \Bar{L}_{t}(\Bar{\ry}_t- \Bar{F}_t \Bar{\rx}_{t|t-1})
\end{align}
where the Kalman gain $\Bar{L}_t$ can be computed recursively with the following equations:
\begin{align}
    \Bar{L}_t &= \Bar{W}_{t|t-1}\Bar{F}_t^{\T}(\Bar{F}_t \Bar{W}_{t|t-1}\Bar{F}_t^{\T}+\Bar{V}_t)^{\dagger}\label{eqn:kalman-gain1}\\
    \Bar{W}_t &= (I-\Bar{L}_t \Bar{F}_t)\Bar{W}_{t|t-1}\label{eqn:kalman-gain2}\\
    \Bar{W}_{t+1|t} &=A_t\Bar{W}_t A_t^{\T} + W_t\label{eqn:kalman-gain3}
\end{align}
Denote $\Bar{N}_t = \bar{L}_{t+1}(\bar{F}_{t+1}\bar{W}_{t+1|t}\bar{F}_{t+1}^{\T}+\bar{V}_{t+1})\bar{L}_{t+1}^{\T}$, which will be used later. Denote the estimation error of $\rx_t$ given $\I_t$ as:
\begin{align*}
    \bar{\rw}_t = \rx_t - \bar{\rx}_t
\end{align*}
From Kalman filtering theory, we can show the following proposition.
\begin{prop} \label{prop:decomposition}
Given $\bar{\ry}_{t}$, the Kalman filtering update equations \eqref{eqn:kalman-update}, \eqref{eqn:kalman-gain1}, \eqref{eqn:kalman-gain2}, \eqref{eqn:kalman-gain3} are true. Moreover, $\bar{\rw}_t \sim\mathcal{N}(0,\bar{W}_t)$ and $\bar{\rw}_t$ is independent of $\I_t$.
\end{prop}
\begin{proof}
We prove this by induction. Firstly for the base case, since $(\rx_1,\bar{\ry}_1)$ are jointly Gaussian, given $\I_1$, we have 
\begin{align*}
    \bar{\rx}_1 = \bar{L}_1 \bar{\ry}_1
\end{align*}
where $\bar{L}_1 = \bar{W}_{1|0}\bar{F}_1^{\T}(\bar{F}_1 \bar{W}_{1|0}\bar{F}^{\T}_1+\bar{V}_1)^{\dagger}$. Similarly, the equations \eqref{eqn:kalman-gain2} and \eqref{eqn:kalman-gain3} and that $\bar{\rw}_1$ is independent of $\I_1$ are true.

Assume these hold for the time step $ t$. We show they also hold for $t+1$. Note that given $\I_t$, conditioning on $\bar{\ry}_{t+1}$ is equivalent to conditioning on $\bar{\ry}_{t+1} - \bar{F}_{t+1}\bar{\rx}_{t+1|t}$ and that $\rx_{t+1} - \bar{\rx}_{t+1|t}$ and $\bar{\ry}_{t+1} - \bar{F}_{t+1}\bar{\rx}_{t+1|t}$ are jointly Gaussian random variables and are independent of $\I_t$ by the induction hypothesis. Hence given $\I_{t+1}$, the $\bar{\rw}_{t+1}$ follows:
\begin{align*}
    \bar{\rw}_{t+1} &= \rx_{t+1} - \bar{\rx}_{t+1|t} -\mathbb{E}[\rx_{t+1} - \bar{\rx}_{t+1|t}|\I_t, \bar{\ry}_{t+1}]\\
    &=\rx_{t+1} - \bar{\rx}_{t+1|t} - \bar{L}_{t+1}(\bar{\ry}_{t+1} - \bar{F}_{t+1}\bar{\rx}_{t+1|t})
\end{align*}
where $\bar{L}_{t+1}$ is computed as \eqref{eqn:kalman-gain1} and $\bar{W}_{t+1}$ is therefore updated by \eqref{eqn:kalman-gain2}. 
Then we show that $\bar{\rw}_{t+1}$ is independent of $\I_{t+1}$. Since, by the assumption, we have that $\rx_{t+1} - \bar{\rx}_{t+1|t}$ and $\bar{\ry}_{t+1} - \bar{F}_{t+1}\bar{\rx}_{t+1|t}$ are jointly Gaussian and jointly independent of $\I_t$, and from the orthogonality principle, we have that $\bar{\rw}_{t+1}$ is independent of $\bar{\ry}_{t+1} - \bar{F}_{t+1}\bar{\rx}_{t+1|t}$, we have
\begin{align*}
    \mathbb{P}(\bar{w}_{t+1} | i_{t+1}) &= \mathbb{P}(\bar{w}_{t+1} | i_{t}, \bar{y}_{t+1})\\
    &= \mathbb{P}(\bar{w}_{t+1} | i_{t}, \bar{y}_{t+1}- \bar{F}_{t+1}\bar{x}_{t+1|t})\\
    &= \mathbb{P}(\bar{w}_{t+1} | \bar{y}_{t+1}- \bar{F}_{t+1}\bar{x}_{t+1|t})\\
    &= \mathbb{P}(\bar{w}_{t+1})
\end{align*}
Hence, $\bar{\rw}_{t+1}$ is independent of $\I_{t+1}$.
This completes the proof.
\end{proof}

Proposition \ref{prop:decomposition} implies that $\rx_t$ can be decomposed as $\bar{\rx}_{t}+\bar{\rw}_t$ and $\bar{\rw}_t$ is independent of $\I_t$. Denote $\rc_t$ as the output of the encoder at time step $t$. Note that $\rc^{t} - \I_t-\rx_t$ is Markov. This implies that $\rc^{t}$ provides zero information about $\bar{\rw}_t$, hence the conditional expectation of $\rx_t$ is the same as the conditional expectation of $\bar{\rx}_t$:
\begin{align}
    \mathbb{E}[\rx_t | \rc^{t}, \ry^{t}] = \mathbb{E}[\bar{\rx}_t | \rc^{t}, \ry^{t}]
\end{align}
This implies that $\bar{\rw}_t$ is the part of $\rx_t$ that can never be estimated from the observations of the encoder and the controller.
We can finally define the \textit{observable Markov process}:
\begin{equation}\label{eqn:observable-process-LQG}
    \Bar{\rx}_{t+1} = A_t \Bar{\rx}_t +B_t \ru_t + \Bar{\rn}_t
\end{equation}
The $\bar{\rn}_t$s are independent Gaussian noises with zero mean and covariance matrix $\bar{N}_t$, which is a direct result of Proposition \ref{prop:decomposition}.

\subsection{Gaussian-Markov Source}
For the Gaussian-Markov source, setting $u_t = 0$, we have
\begin{align*}
    \bar{\rx}_t = \mathbb{E}[\rx_t | \ry^{t},\rz^{t}] = \mathbb{E}[\rx_t |\I_t]
\end{align*}
It is immediate from the orthogonality principle that $\bar{\rw}_t = \rx_t - \bar{\rx}_t$ is independent of $\I_t$ and is a Gaussian random variable with zero mean and covariance matrix $\bar{W}_t$. The update Kalman filter equations also follow from \eqref{eqn:kalman-update}, \eqref{eqn:kalman-gain1}, \eqref{eqn:kalman-gain2}, and \eqref{eqn:kalman-gain3}. The \textit{observable Markov process} for Gaussian-Markov source is:
\begin{equation}\label{eqn:observable-process-GM}
    \Bar{\rx}_{t+1} = A_t \Bar{\rx}_t + \Bar{\rn}_t
\end{equation}
where $\bar{\rn}_t \sim \mathcal{N}(0,\bar{N}_t)$ are independent Gaussian noises as well.

\section{Sufficient Linear Policy}

In this section, we consider the LQG plant and show that it is sufficient to consider the linear policy set $\Omega_1$ as in Definition \ref{def:LQG-opt-policy} for the optimization problem \eqref{LQG-plain-opt-form}. The same argument holds for the Gaussian-Markov source and hence is omitted here. 

The policy set $\Omega_1$ consists of a linear encoder and a certainty equivalence controller.
The optimal control gain $K_t$ of the certainty equivalence controller can be computed via the following backward Riccati recursions.
\begin{align}
    S_{t} &= \begin{cases} 
      Q_t & \text{if } t=T \\
      Q_t + \Phi_{t+1}  & \text{if } t<T
   \end{cases}\\
   \Phi_t &= A_t^{\T}(S_t - S_tB_t(B_t^{\T}S_tB_t+R_t)^{-1}B_t^{\T}S_t)A_t\\
   K_t &= -(B_t^{\T}S_t B_t +R_t)^{-1}B_t^{\T}S_t A_t\\
   \Theta_t &= K_t^{\T}(B_t^{\T}S_t B_t+R_t)K_t
\end{align}
As in Proposition \ref{prop:decomposition}, $\mathbb{E}[\rx_t|\ry^{t},\rz^{t},\ru^{t-1}]$ is a linear combination of $(\ry^{t},\rz^{t},\ru^{t-1})$. Denote the coefficient of $\mathbb{E}[\rx_t|\ry^{t},\rz^{t},\ru^{t-1}]$ for $\rz^{t}$ as $L_{t}^{*}$.
The policy set $\Omega_1$ is defined as:
\begin{defn}[Policy Set $\Omega_1$]\label{def:LQG-opt-policy}
    The policy set $\Omega_1$ consists of the following linear encoders and a certainty-equivalence controller with control gain $K_{t}$:
    \begin{enumerate}
        \item A linear encoder:
        \begin{align}
            \rc_t = L'_t \rz^{t} + \rh_t
        \end{align}
        where $L'_{t} = L_t L^{*}_t$ and $\rh_t\sim \mathcal{N}(0,H_{t})$ is white Gaussian noise independent of $(\rx^{t},\ry^{t},\rz^{t},\ru^{t-1})$ with $H_{t}\succeq 0$;
        \item The certainty equivalence controller $\ru_t = K_t \mathbb{E}[\rx_t|\ry^t , \rc^t, \ru^{t-1}]$.
    \end{enumerate}
\end{defn}
\begin{remark}\label{remark:linear-controller}
    Since $(y^{t},u^{t-1})$ are known to the controller and $\bar{x}_t = \mathbb{E}[x_t|y^{t},z^{t},u^{t-1}]$, effectively, the codeword $c_t$ carries the same information as $L_t \bar{x}_t + h_t$ at the controller side.
\end{remark}

\begin{thm}[Sufficient Linear Policy]\label{thm:linear-policy-opt}
It is sufficient to consider the policy set $\Omega_1$ for the optimization problem \eqref{LQG-plain-opt-form}, and \eqref{LQG-plain-opt-form} is equivalent to 
    \begin{align*}
        \inf_{\omega\in\Omega_1, J\leq d} \sum_{t=1}^{T}I(\Bar{\rx}_t; \rc_t |\rc^{t-1},\ry^{t})
    \end{align*}
\end{thm}
\begin{proof}
See Section \ref{proof:linear-policy-opt}.
\end{proof} 
\begin{remark}
    The sufficient linear policy structure recovers the corresponding structural results from the past works \cite{GMsdp,lqgsdp,CCKS2022,LQG-sideinfo1,LQG-sideinfo2,partial-observation-GM}, which correspond to the cases of full observation with or without linear side information and partial observations without side information for both LQG plants and Gaussian-Markov sources.
\end{remark}

\section{Convex Formulation for Scalar Case}

In this section, we consider scalar time-varying and time-invariant systems. We show that the resulting optimization problem in Theorem \ref{thm:linear-policy-opt} admits convex formulations for finite-horizon time-varying and infinite-horizon time-invariant systems.

\subsection{Time-Varying System}
In this section, we consider time-varying systems. We assume $W_t>0$, $C_t \neq 0$, $N_t>0$, $F_t\neq 0 $ and $V_t>0$. We first define some error covariance matrices and find some relationships between them from the standard Kalman filtering equations. Denote the following error covariance matrices:
\begin{align*}
    \bar{P}_{t|t} &= \Cov\left(\bar{\rx}_{t} - \mathbb{E}[\bar{\rx}_t | \rc^{t},\ry^{t}]\right)\\
    \bar{P}^{\SI}_{t+1|t}&=\Cov\left(\bar{\rx}_{t+1} - \mathbb{E}[\bar{\rx}_{t+1} | \rc^{t},\ry^{t+1}]\right)\\
    P_{t|t} &= \Cov\left(\rx_t - \mathbb{E}[\rx_t|\rc^{t},\ry^{t}]\right)\\
    P_{t+1|t}^{\SI} &= \Cov\left(\rx_{t+1}-\mathbb{E}[\rx_{t+1}|\rc^{t},\ry^{t+1}]\right)
\end{align*}
where the superscript SI indicates the error covariance after observing the side information.
Since $\rx_t = \bar{\rx}_t + \bar{\rw}_t$ and $\bar{\rw}_t$ is independent of $(\ry^{t},\rz^{t},\ru^{t})$, the following relationships are true between $\bar{P}$ and $P$ due to Proposition \ref{prop:decomposition}:
\begin{align*}
    P_{t|t} &= \bar{P}_{t|t} + \bar{W}_t\\
    P_{t+1|t} &= \bar{P}_{t+1|t} +\bar{W}_{t+1}\\
    P_{t+1|t}^{\SI} &= \bar{P}_{t+1|t}^{\SI} + \bar{W}_{t+1} 
\end{align*}
These equations are useful to find the relationship between $\bar{P}_{t|t}$ and $\bar{P}_{t+1|t}^{\SI}$. The update from $P_{t+1|t}$ to $P_{t+1|t}^{\SI}$ is
\begin{align*}
    (P_{t+1|t}^{\SI})^{-1} = P_{t+1|t}^{-1} + F_{t+1}^{\T}V_{t+1}^{-1}F_{t+1}
\end{align*}
Therefore, by the equations between $\bar{P}$ and $P$, we have the following equalities:
\begin{align}\label{eqn:y1update-for-Pbar}
    (P_{t+1|t}^{\SI})^{-1} &= (\bar{P}_{t+1|t}^{\SI}+\bar{W}_{t+1})^{-1}\nonumber\\
                        &= (\bar{P}_{t+1|t} + \bar{W}_{t+1})^{-1} + F_{t+1}^{\T}V_{t+1}^{-1}F_{t+1}
\end{align}
Since $\bar{P}_{t+1|t} = A_{t}\bar{P}_{t|t}A_{t}^{\T}+\bar{N}_t$, the equalities above present the update from $\bar{P}_{t|t}$ to $\bar{P}_{t+1|t}^{\SI}$.


 To present the convex optimization form of \eqref{LQG-plain-opt-form} for the scalar case, we first denote the following variables
\begin{align*}
    \alpha_{t} &= A_{t}^{2}(V_{t+1}-F_{t+1}^{2}\bar{W}_{t+1})\\
    \beta_t &= \bar{N}_t(V_{t+1} - F_{t+1}^{2}\bar{W}_{t+1})-F_{t+1}^{2}\bar{W}_{t+1}^{2}\\
    \gamma_t &= A_{t}^{2}F_{t+1}^{2}\\
    \delta_t &= V_{t+1}+F_{t+1}^{2}(\bar{N}_{t}+\bar{W}_{t+1})
\end{align*} and the function $f_{t}(\bar{P}_{t|t})$
\begin{align*}
    f_{t}(\bar{P}_{t|t}) = \log(\alpha_t  \bar{P}_{t|t}+\beta_t) - \log(\gamma_t \bar{P}_{t|t} + \delta_t)-\log\bar{P}_{t|t}
\end{align*}

\begin{thm}[Convex-TV]\label{thm:convex-time-varying}
For the scalar case, if $W_t>0$, $C_t \neq 0$, $N_t>0$, $F_t\neq 0 $ and $V_t>0$, \eqref{LQG-plain-opt-form} equals $\max\{0,\eqref{convex-TV}\}$.
\begin{subequations}\label{convex-TV}
\begin{align}
    \inf_{\{\bar{P}_{t|t}\}_{t\in[T]}} & \sum_{t=1}^{T-1} \frac{1}{2}f_{t}(\bar{P}_{t|t}) + \frac{1}{2}(\log\bar{P}_{1|0}^{\SI} - \log\bar{P}_{T|T})\\
    \text{s.t. }& \sum_{t=1}^{T}\Tr(\Theta_t \bar{P}_{t|t}) + \eta_1 \leq Td\\
    & \begin{bmatrix}
       \bar{P}_{t|t-1}  - \bar{P}_{t|t}& (\bar{P}_{t|t-1}+\bar{W}_t)F_{t} \\
        F_{t}(\bar{P}_{t|t-1}+\bar{W}_t) & F_{t}^{2}(\bar{P}_{t|t-1}+\bar{W}_t)+V_t
    \end{bmatrix} \succeq 0 \\
        & \bar{P}_{t+1|t} = A_{t}\bar{P}_{t|t} A_{t}^{\T} + \bar{N}_t \\
    &\bar{P}_{t|t} > 0
\end{align}
\end{subequations}
where $\eta_1 = \sum_{t=1}^{T}[\Tr(S_t W_t) + \Tr(\Theta_t \bar{W}_t)] + \Tr(\Phi_1 P_{1|0})$
    
\end{thm}

\begin{proof}
    See Section \ref{proof:convex-TV}.
\end{proof}

Note that the effective codeword $\rc_t = L_t \bar{\rx}_t + \rh_t$ can be found by choosing $L_t$ and $H_t$ such that
\begin{align*}
    \bar{P}_{t|t}^{-1} - (\bar{P}_{t|t-1}^{\SI})^{-1} = L_{t}^{2}H_t^{-1}
\end{align*}

\subsection{Time-Invariant System}\label{sec:LQG-TI}

In this section, we consider time-invariant systems, i.e., the system parameters, like $A_t = A$, $B_t = B$, $W_t = W$, $C_t = C$ are all the same across time steps. We assume $W>0$, $C\neq 0$, $N>0$, $F\neq 0 $ and $V>0$. We also consider $T\rightarrow \infty$. The optimization problem \eqref{LQG-plain-opt-form} is replaced by the following:
\begin{align}\label{LQG-plain-opt-form-TI}
    \inf_{\omega\in\Omega, J\leq d} \limsup_{T\rightarrow \infty}\frac{1}{T}I(\rz^{T}\rightarrow \ru^T \Vert \ry^T)
\end{align}
where $J  = \limsup_{T\rightarrow \infty}\frac{1}{T}\sum_{t=1}^{T}\mathbb{E}[\Vert \rx_{t+1}\Vert_{Q}^{2} + \Vert \ru_t \Vert_{R}^{2}]$.

We assume that $(A,B)$ is stabilizable.
Since $Q\succ 0$ and $R\succ 0$, the control Riccati recursion converges, and hence the optimal control gain matrix converges to its steady-state value \cite{kailath2000linear}. 
The optimal control gain matrix $K$ in the infinite-horizon case can be computed by first finding the unique positive definite solution to the following Riccati equation:
\begin{align}
    A^{\T}SA-S-A^{\T}SB(B^{\T}SB+R)^{-1}B^{\T}SA+Q = 0
\end{align}
Then the control gain matrix $K = -(B^{\T}SB+R)^{-1}B^{\T}SA$. Set $\Theta = K^{\T}(B^{\T}SB+R)K$.

Since $(A,\bar{F})$ is detectable and $W\succ 0$, the error covariance matrix of estimating $\rx_t$ given $\I_t$, $\bar{W}_t$, converges to the solution of the following discrete algebraic Riccati equation by \cite{kailath2000linear}:
\begin{align}
    \bar{W} = \bar{W}^{+} -\bar{W}^{+}\bar{F}^{\T}(\bar{V}+\bar{F}\bar{W}^{+}\bar{F}^{\T})^{-1}\bar{F}\bar{W}^{+}
\end{align}
where $\bar{W}^{+} = A\bar{W}A^{\T}+W$. Since $\bar{N}_t = \bar{W}_{t+1|t} - \bar{W}_{t+1}$, by \eqref{eqn:kalman-gain1} and \eqref{eqn:kalman-gain2}, $\bar{N}_t$ converges to $\bar{N}\triangleq\bar{W}^{+} - \bar{W}$. Immediately, $\alpha_t$, $\beta_t$, $\gamma_t$ and $\delta_t$ converge to the following terms, respectively:
\begin{align*}
    \alpha &= A^{2}(V-F^{2}\bar{W})\\
    \beta &= \bar{N}(V-F^{2}\bar{W}) - F^{2}\bar{W}^{2}\\
    \gamma &= A^{2}F^{2} \\
    \delta &= V+F^{2}(\bar{N}+\bar{W})
\end{align*}
These are all non-negative parameters following the proof in \ref{proof:convex-TV}.
Denote the function $f(\bar{P}_{t|t}) $ as:
\begin{align*}
    f(\bar{P}_{t|t}) = \log(\alpha  \bar{P}_{t|t}+\beta) - \log(\gamma \bar{P}_{t|t} + \delta)-\log\bar{P}_{t|t}
\end{align*} 
We have the following single-letter convex optimization form for \eqref{LQG-plain-opt-form-TI}.

\begin{thm}[Convex-TI]\label{thm:convex-TI}
    If $W>0$, $C\neq 0$, $N>0$, $F\neq 0 $ and $V>0$, \eqref{LQG-plain-opt-form-TI} equals $\max\{0,\eqref{convex-TI}\}$.
    \begin{subequations}\label{convex-TI}
    \begin{align}
    \inf_{\bar{P}> 0} & \frac{1}{2}\left[\log(\alpha  \bar{P}+\beta) - \log(\gamma \bar{P} + \delta)-\log\bar{P}\right] \\
    \text{s.t. }&  \Tr(\Theta \bar{P})+\Tr(S W)+ \Tr(\Theta \bar{W}) \leq d \\
    & \begin{bmatrix}
        \bar{P}^{+}-\bar{P} & (\bar{P}^{+}+\bar{W})F\\
        F(\bar{P}^{+}+\bar{W}) & F^{2}(\bar{P}^{+}+\bar{W})+V
    \end{bmatrix} \succeq 0 \\
        & \bar{P}^{+} = A\Bar{P} A^{\T} + \bar{N} 
\end{align}
\end{subequations}
\end{thm}

\begin{proof}
    See Section \ref{proof:convex-TI}.
\end{proof}

We end this section with the closed-form solution of the above single-letter convex optimization problem.

\begin{corollary}\label{col:explicit-solution-TI}
    For $|A| >1$, the optimal value of the convex optimization \eqref{convex-TI} is
    \begin{align*}
        &\frac{1}{2}\log\left(A^{2}(1-F^{2}V^{-1}\bar{W})+\eta_{0}\right)\\
        -&\frac{1}{2}\log\left(A^{2}F^{2}V^{-1}\frac{d-SW}{\Theta}+1+F^{2}V^{-1}W\right)
    \end{align*}
    when $\Theta \bar{W} +SW < d \leq \Theta \bar{W } + SW + \Theta \bar{P}^{*}$, where 
    \begin{align*}
        \eta_0 = \frac{\bar{N}(1-F^{2}V^{-1}\bar{W})-F^{2}V^{-1}\bar{W}^{2}}{\frac{d-SW}{\Theta}-\bar{W}}
    \end{align*}
    The optimal value is 0 when $d> \Theta \bar{W } + SW + \Theta \bar{P}^{*}$. The $\bar{P}^{*}$ is the positive solution to the following polynomial
    \begin{align*}
        -(\bar{P}+\bar{W})^{2} F^{2}A^{2} + (\bar{P}+\bar{W})\left[(A^{2}-1)V-F^{2} W \right]&\\
        +VW &= 0
    \end{align*}
\end{corollary}
\begin{proof}
    See Section \ref{proof:explicit-sol-TI}.
\end{proof}

\begin{remark}
    In the case of full observation and side information, we have $\bar{W} = 0$ and hence $\bar{N} = W$. Substituting these values recovers the explicit solution in \cite{LQG-sideinfo1}.
\end{remark}

\section{Gaussian-Markov Process}\label{sec:G-M-results}
In this section, we consider Gaussian-Markov sources with the weighted mean-square error distortion. Similar to the LQG plant, it is also sufficient to consider a set of linear policies for the causally conditioned directed information optimization problem. The convex optimization forms also hold for both the time-varying systems and the time-invariant systems. The only difference is that the control cost constraint is replaced by the weighted mean-square error distortion constraint.

As in Proposition \ref{prop:decomposition}, $\mathbb{E}[\rx_t|\ry^{t},\rz^{t}]$ is a linear combination of $(\ry^{t},\rz^{t})$. Denote the coefficient of $\mathbb{E}[\rx_t|\ry^{t},\rz^{t}]$ for $\rz^{t}$ as $L_{t}^{*}$.
The optimal linear policy set for the optimization problem \eqref{GM-plain-opt-form} is denoted as the policy set $\Lambda_1$:
\begin{defn}[Policy Set $\Lambda_1$]\label{def:GM-opt-policy}
    The policy set $\Lambda_1$ consists of linear encoders and the conditional mean decoder:
    \begin{enumerate}
         \item A linear encoder:
        \begin{align}
            \rc_t = L'_{t} \rz^{t} + \rh_t
        \end{align}
        where $L'_{t} = L_t L^{*}_t$ and $\rh_t\sim \mathcal{N}(0,H_{t})$ is white Gaussian noise independent of $(\rx^{t},\ry^{t},\rz^{t})$ with $H_{t}\succeq 0$;
        \item The conditional mean decoder $
        \hat{\rx}_t = \mathbb{E}[\rx_t|\ry^t , \rc^t]$.
    \end{enumerate}
\end{defn}
\begin{remark}
    Since $y^{t}$ are known to the decoder and $\bar{x}_t = \mathbb{E}[x_t|y^{t},z^{t}]$, effectively, the codeword $c_t$ carries the same information as $L_t \bar{x}_t + h_t$ at the decoder.
\end{remark}
Following the proof of Theorem \ref{thm:linear-policy-opt}, we have the following Corollary for the sufficient linear policy set of the Gaussian-Markov sources.
\begin{corollary}\label{cor:optimal-policy-GM}
It is sufficient to consider the linear policy set $\Lambda_1$ for the optimization problem \eqref{GM-plain-opt-form}, and \eqref{GM-plain-opt-form} is equivalent to 
    \begin{align*}
        \inf_{\lambda\in\Lambda_1, \text{WMSE}\leq d} \sum_{t=1}^{T}I(\Bar{\rx}_t; \rc_t |\rc^{t-1},\ry^{t})
    \end{align*}
\end{corollary}

The proof follows from that for Theorem \ref{thm:linear-policy-opt}, except that $\ru_t$ is replaced by $\hat{\rx}_t$.
As a direct result, for the scalar case, the optimization problem \eqref{GM-plain-opt-form} admits the same convex formulation as in Theorem \ref{thm:convex-time-varying} except that the control cost constraint is replaced by the following weighted mean-square error distortion constraint:
\begin{align*}
    \sum_{t=1}^{T}[\Tr(\Xi_t \bar{P}_{t|t})+\Tr(\Xi_t \bar{W}_t)] \leq Td
\end{align*}

For the time-invariant systems, we consider the following optimization problem
\begin{align}\label{GM-plain-opt-TI}
    \inf_{\lambda\in\Lambda, \text{WMSE}\leq d} \limsup_{T\rightarrow \infty}\frac{1}{T}I(\rz^{T}\rightarrow \hat{\rx}^T \Vert \ry^T)
\end{align}
where $\text{WMSE} = \limsup_{T\rightarrow \infty}\frac{1}{T}\sum_{t=1}^{T}\mathbb{E}[\Vert\rx_t - \hat{\rx}_t\Vert^{2}_{\Xi}]$.
Following the proof of Theorem \ref{thm:convex-TI}, the optimization problem \eqref{GM-plain-opt-TI} is equivalent to the single-letter form as in Theorem \ref{thm:convex-TI}, except that the control cost constraint is replaced by the following:
\begin{align*}
    \Tr(\Xi\bar{P}) + \Tr(\Xi\bar{W})\leq d
\end{align*}

\section{Numerical Simulation}
In this section, we present numerical simulations for the time-invariant scalar systems. We draw the SRDF curves for four cases: full observation, partial observation, full observation with side information, and partial observation with side information. These four cases correspond to the dash-dot line, dotted line, dashed line, and solid line, respectively.

In particular, we consider the LQG plant with $B=W=Q=R=F=C = 1$ and $A = 5$. We adjust the values of $N$ and $V$ to indicate different qualities of partial observations and side information. For the case with the same quality of partial observation and side information, we set $N=V=1$, as shown in Figure \ref{fig:simulation}. 
For the case with high-quality partial observation and low-quality side information, we set $N=0.1$ and $V=1$, as shown in Figure \ref{fig:simulation1}. For the case with low-quality partial observation and high-quality side information, we set $N=1$ and $V=0.1$, as shown in Figure \ref{fig:simulation2}.

\begin{figure*}[t]
    \centering
    \begin{subfigure}[t]{0.32\textwidth}
        \centering
        \includegraphics[width=\linewidth]{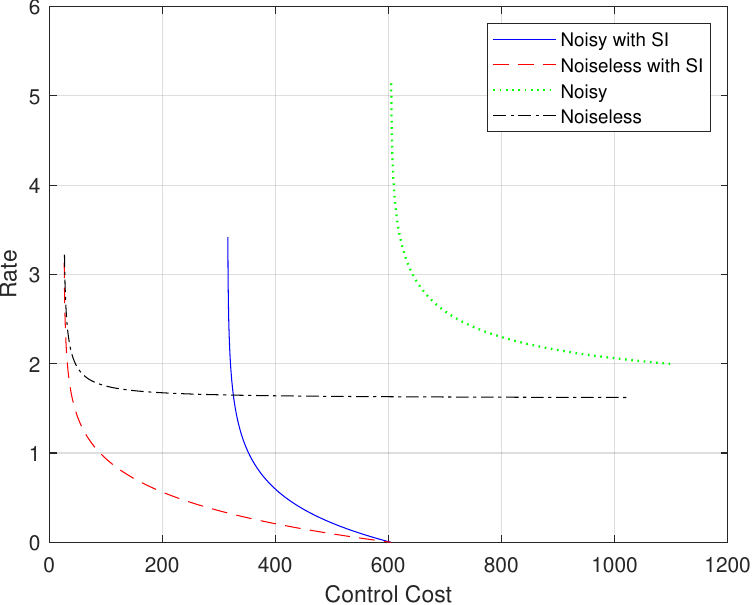}
        \caption{$N=V=1$. Equal-quality partial observation and side information.}
        \label{fig:simulation}
    \end{subfigure}
    \hfill
    \begin{subfigure}[t]{0.32\textwidth}
        \centering
        \includegraphics[width=\linewidth]{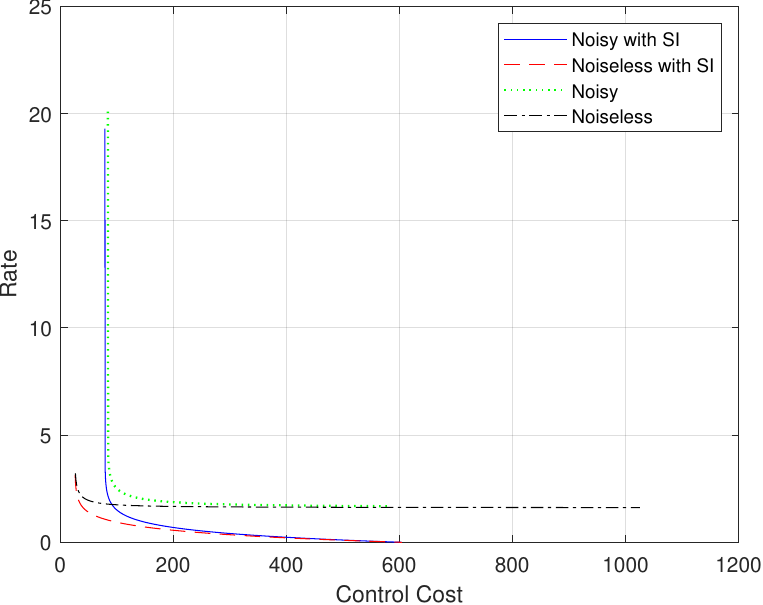}
        \caption{$N=0.1$, $V=1$. High-quality partial observation and low-quality side information.}
        \label{fig:simulation1}
    \end{subfigure}
    \hfill
    \begin{subfigure}[t]{0.32\textwidth}
        \centering
        \includegraphics[width=\linewidth]{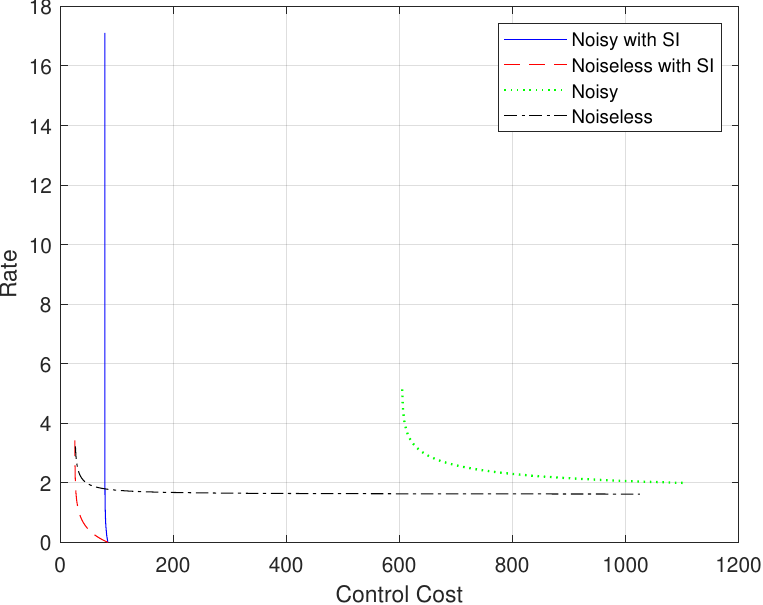}
        \caption{$N=1$, $V=0.1$. Low-quality partial observation and high-quality side information.}
        \label{fig:simulation2}
    \end{subfigure}

    \caption{Rate-cost tradeoff curves for the four settings with different qualities of partial observations and side information.}
    \label{fig:three_cases}
\end{figure*}

The numerical results lead to the following observations. First,  side information can reduce the minimal amount of information required even when its quality is lower than that of the partial observation. Second, side information can fundamentally change the partial observation setting by enabling control costs that would otherwise be unattainable using the partial observation alone. This is particularly pronounced when the quality of side information is comparable to that of the partial observation.


\section{Lemmas and Proofs}

\subsection{Lemmas}

\begin{lemma}\label{lemma:P-equals-G}
    For the LQG plant, we have
    \begin{align*}
        \mathbb{P}(\bar{x}_{t+1},y_{t+1}|u^{t},y^{t},\bar{x}_t) = \mathbb{G}(\bar{x}_{t+1},y_{t+1}|u^{t},y^{t},\bar{x}_t)
    \end{align*}
\end{lemma}
\begin{proof}
    By the observable process \eqref{eqn:observable-process-LQG} and Proposition \ref{prop:decomposition}, we have that given $(\bar{\rx}_t,\ru_t)$, $\bar{\rx}_{t+1}$ is left with $\bar{\rn}_t$, which is a Gaussian random variable independent of $\I_t$ and hence independent of $(\ru^t, \ry^t, \bar{\rx}_t)$, which implies:
    \begin{align*}
        \mathbb{P}(\bar{\rx}_{t+1}|u^{t},y^{t},\bar{x}_t) 
        = \mathbb{P}(A_t\bar{x}_t+B_t u_t+\bar{\rn}_t)
    \end{align*}
    Similarly, for $\ry_{t+1}$, we have:
    \begin{align*}
        &\mathbb{P}(\ry_{t+1}|u^{t},y^{t},\bar{x}_t,\bar{x}_{t+1})\\
       =&\mathbb{P}(F_{t+1}\bar{x}_{t+1}+F_{t+1}\bar{\rw}_{t+1}+\rv_{t+1}|\bar{n}_{t})
    \end{align*}
    where the equality is by $\ry_{t+1} = F_{t+1}\rx_{t+1}+\rv_{t+1}$, $\bar{\rw}_{t+1} = \rx_{t+1}-\bar{\rx}_{t+1}$, the observable process \eqref{eqn:observable-process-LQG}, and the fact that $(\bar{\rw}_{t+1},\rv_{t+1},\bar{\rn}_t)$ are jointly independent of $(\ru^{t},\ry^{t},\bar{\rx}_t)$.

    Next, we show that the Gaussian version $\mathbb{G}$ follows the same distribution.
     Since $(\ry^{t},\ru^{t-1})-(\ru_t,\bar{\rx}_t)-\bar{\rx}_{t+1}$ is Markov and follows the observable process \eqref{eqn:observable-process-LQG}, by \cite[Lemma 5]{lqgsdp-arxiv}, the same Markov chain also holds in $\mathbb{G}$. Similarly, the Markov chain $(\ru^{t-1},\ry^{t})-(\ru_t,\bar{\rx}_t,\bar{\rx}_{t+1})-\ry_{t+1}$ also holds in $\mathbb{G}$. By \cite[Lemma 6]{lqgsdp-arxiv}, the linear relationship \eqref{eqn:observable-process-LQG} also holds in $\mathbb{G}$. Hence we have
    \begin{align*}
        \mathbb{G}(\bar{\rx}_{t+1}|u^{t},y^{t},\bar{x}_t)
        =& \mathbb{G}(\bar{\rx}_{t+1}|u_t,\bar{x}_t)\\
        =&\mathbb{P}(\bar{\rx}_{t+1}|u_t,\bar{x}_t)\\
        =&\mathbb{P}(A_t\bar{x}_t+B_tu_t+\bar{\rn}_t)\\
        =&\mathbb{P}(\bar{\rx}_{t+1}|u^{t},y^{t},\bar{x}_t)
    \end{align*}
    
    Similarly, since $(\bar{\rw}_{t+1},\rv_{t+1},\bar{\rn}_t)$ are jointly Gaussian and independent of $(\ru^{t},\ry^{t},\bar{\rx}_t)$ in $\mathbb{P}$, by \cite[Lemma 5]{lqgsdp-arxiv}, $(F_{t+1}\bar{\rw}_{t+1}+\rv_{t+1})-\bar{\rn}_t -(\ru^{t},\ry^{t},\bar{\rx}_t)$ is also Markov in $\mathbb{G}$. Therefore, we have:
    \begin{align*}
        &\mathbb{G}(\ry_{t+1}|u^{t},y^{t},\bar{x}_t,\bar{x}_{t+1}) \\
        =&\mathbb{G}(F_{t+1}\bar{x}_{t+1}+F_{t+1}\bar{\rw}_{t+1}+\rv_{t+1}|u^{t},y^{t},\bar{x}_t,\bar{n}_{t})\\
        =&\mathbb{G}(F_{t+1}\bar{x}_{t+1}+F_{t+1}\bar{\rw}_{t+1}+\rv_{t+1}|\bar{n}_{t})\\
        =&\mathbb{P}(F_{t+1}\bar{x}_{t+1}+F_{t+1}\bar{\rw}_{t+1}+\rv_{t+1}|\bar{n}_{t})\\
        =& \mathbb{P}(\ry_{t+1}|u^{t},y^{t},\bar{x}_t,\bar{x}_{t+1})
    \end{align*}
Combining these two groups of equalities, we show that these two conditional distributions are the same in $\mathbb{P}$ and $\mathbb{G}$.
 
\end{proof}

\begin{lemma}\label{lemma:P-minus-G}
Let $\mathbb{P}(\bar{\rx}^{T+1},\ru^{T},\ry^{T+1})$ be the measure generated by a policy $\omega\in \Omega$, \eqref{eqn:hidden-process}, \eqref{eqn:enc-observation}, and \eqref{eqn:side-info}. Let $\mathbb{G}(\bar{\rx}^{T+1},\ru^{T},\ry^{T+1})$ be a Gaussian measure with the same first and second moments as $\mathbb{P}$. Suppose $I_{\mathbb{P}}(\bar{\rx}_t;\ru_t|\ru^{t-1},\ry^{t})$ is finite for each time step. Then we have

\noindent (a) For $t\in[T]$, we have
    \begin{align*}
        &I_{\mathbb{P}}(\bar{\rx}_t;\ru_t|\ru^{t-1},\ry^{t})- I_{\mathbb{G}}(\bar{\rx}_t;\ru_t|\ru^{t-1},\ry^{t})\\
        =&\KL_{\mathbb{P},\mathbb{G}}(\bar{\rx}_{t}|\ru^{t},\ry^{t})-\KL_{\mathbb{P},\mathbb{G}}(\bar{\rx}_{t}|\ru^{t-1},\ry^{t})
    \end{align*}

\noindent (b) For $t\in[T]$, we have
\begin{align*}
    &\KL_{\mathbb{P},\mathbb{G}}(\bar{\rx}_{t}|\ru^{t},\ry^{t})\\
    =&\KL_{\mathbb{P},\mathbb{G}}(\bar{\rx}_{t+1}|\ru^{t},\ry^{t+1}) + \KL_{\mathbb{P},\mathbb{G}}(\ry_{t+1}|\ru^{t},\ry^{t})\\
    &+\KL_{\mathbb{P},\mathbb{G}}(\bar{\rx}_{t}|\ru^{t},\ry^{t+1},\bar{\rx}_{t+1})
\end{align*}

\end{lemma}
\begin{proof}
We show this by induction. For $t=1$, since the initial distribution $\rx_1$ is Gaussian, we have $\KL_{\mathbb{P},\mathbb{G}}(\bar{\rx}_1|\ry^{1}) = 0$. Since $I_{\mathbb{P}}(\bar{\rx}_1;\ru_1|\ry^{1})$ is finite, by \cite[Lemma 6]{longversion}, (a) holds for $t=1$ and $\KL_{\mathbb{P},\mathbb{G}}(\bar{\rx}_{1}|\ru^{1},\ry^{1})$ is finite. The finiteness implies 
\begin{align*}
    \mathbb{P}_{\bar{\rx}_1|u^{1},y^{1}} \ll \mathbb{G}_{\bar{\rx}_1|u^{1},y^{1}}
\end{align*}
holds $\mathbb{P}(\ru^{1},\ry^{1})$ a.s. By the absolute continuity together with Lemma \ref{lemma:P-equals-G}, we have 
\begin{align*}
    \mathbb{P}_{\bar{\rx}^{2},\ry_{2}|u^{1},y^{1}} \ll \mathbb{G}_{\bar{\rx}^{2},\ry_{2}|u^{1},y^{1}}
\end{align*}
holds $\mathbb{P}(\ru^{1},\ry^{1})$ a.s., and
\begin{align*}
    &\KL_{\mathbb{P},\mathbb{G}}(\bar{\rx}_{1}|\ru^{1},\ry^{1})\\
    =& \KL_{\mathbb{P},\mathbb{G}}(\bar{\rx}^{2},\ry_{2}|\ru^{1},\ry^{1})\\
    =&\KL_{\mathbb{P},\mathbb{G}}(\bar{\rx}_{2}|\ru^{1},\ry^{2}) + \KL_{\mathbb{P},\mathbb{G}}(\ry_{2}|\ru^{1},\ry^{1})\\
    &+\KL_{\mathbb{P},\mathbb{G}}(\bar{\rx}_{1}|\ru^{1},\ry^{2},\bar{\rx}_{2}).
\end{align*}
This completes the proof for $t=1$. Suppose (a) and (b) hold up to time step $t$. We next show they also hold for $t+1$. Since (b) holds at time step $t$, we have $\KL_{\mathbb{P},\mathbb{G}}(\bar{\rx}_{t+1}|\ru^{t},\ry^{t+1})$ is finite. Combining the finiteness of $I_{\mathbb{P}}(\bar{\rx}_{t+1};\ru_{t+1}|\ru^{t},\ry^{t+1})$ and \cite[Lemma 6]{longversion}, we have that (a) holds for time step $t+1$ and $\KL_{\mathbb{P},\mathbb{G}}(\bar{\rx}_{t+1}|\ru^{t+1},\ry^{t+1})$ is finite. The finiteness implies 
\begin{align*}
    \mathbb{P}_{\bar{\rx}_{t+1}|u^{t+1},y^{t+1}}\ll \mathbb{G}_{\bar{\rx}_{t+1}|u^{t+1},y^{t+1}}
\end{align*}
holds $\mathbb{P}(\ru^{t+1},\ry^{t+1})$ a.s. Following the same argument for $t=1$, the absolute continuity here together with Lemma \ref{lemma:P-equals-G} implies that (b) holds for $t+1$. This completes the induction.

\end{proof}

\subsection{Proof of Theorem \ref{thm:linear-policy-opt}}\label{proof:linear-policy-opt}

    The proof follows a structure similar to that of \cite{lqgsdp}. For any feasible policy in $\Omega$, we denote the corresponding probability measure of $(\ry^{T+1},\rz^{T+1},\ru^{T},\rx^{T+1})$ as $\mathbb{P}$ and its Gaussian version, which has the same first and second moments, as $\mathbb{G}$.
    We show Theorem \ref{thm:linear-policy-opt} by establishing the following inequalities: 
    \begin{align}
        &\inf_{\omega\in\Omega, J\leq d} I(\rz^{T}\rightarrow \ru^T \Vert \ry^T) \nonumber\\
        \geq &\inf_{\omega\in\Omega, J\leq d} \sum_{t=1}^{T}I_{\mathbb{G}}(\Bar{\rx}_t; \ru_t |\ru^{t-1},\ry^{t})\label{ineq1}\\
        \geq &\inf_{\omega\in\Omega_1, J\leq d} \sum_{t=1}^{T}I(\Bar{\rx}_t; \rc_t |\rc^{t-1},\ry^{t})\label{ineq2}\\
        \geq &\inf_{\omega\in\Omega, J\leq d} I(\rz^{T}\rightarrow \ru^T \Vert \ry^T)\label{ineq3}
    \end{align}

    \textit{Proof of inequality $\eqref{ineq1}$:}
    We first establish the following chain of inequalities: 
    \begin{align*}
       I(\rz^{T}\rightarrow \ru^T \Vert \ry^T)
        \geq  \sum_{t=1}^{T}I(\Bar{\rx}_t; \ru_t |\ru^{t-1},\ry^{t})
    \end{align*}
    where the inequality follows from the fact that $\bar{\rx}_t = \mathbb{E}[\rx_t|\I_t]$ and therefore
    \begin{align*}
        I(\rz^{T}\rightarrow \ru^T \Vert \ry^T)
        =&\sum_{t=1}^{T}I(\rz^{t};\ru_t|\ru^{t-1},\ry^{t})\\
        =&\sum_{t=1}^{T}I(\rz^{t},\ry^{t},\ru^{t-1},\bar{\rx}_t;\ru_t|\ru^{t-1},\ry^{t})\\
        \geq & \sum_{t=1}^{T}I(\bar{\rx}_t;\ru_t|\ru^{t-1},\ry^{t}),
    \end{align*}
    This implies that the total rate is lower bounded by the sum of conditional mutual information between the control signal and the observable process $\bar{\rx}_t$.

    Next, we argue that for any feasible policy in $\Omega$ and the induced probability measure $\mathbb{P}$, the objective is lower bounded by its Gaussian version $\mathbb{G}$.
    \begin{lemma}\label{lemma:Gaussian-lower-bound-CDI}
        For any feasible policy and the induced measure $\mathbb{P}$, the causally conditioned directed information is lower bounded by its Gaussian version:
        \begin{align*}
            \sum_{t=1}^{T}I_{\mathbb{P}}(\bar{\rx}_t;\ru_t|\ru^{t-1},\ry^{t}) \geq \sum_{t=1}^{T}I_{\mathbb{G}}(\bar{\rx}_t;\ru_t|\ru^{t-1},\ry^{t})
        \end{align*}
    \end{lemma}
    \begin{proof}
        
From Lemma \ref{lemma:P-minus-G}, we have 
\begin{align*}
    &\sum_{t=1}^{T}\left[I_{\mathbb{P}}(\bar{\rx}_t;\ru_t|\ru^{t-1},\ry^{t}) - I_{\mathbb{G}}(\bar{\rx}_t;\ru_t|\ru^{t-1},\ry^{t})\right]\\
    =&\sum_{t=1}^{T}\left[\KL_{\mathbb{P},\mathbb{G}}(\bar{\rx}_{t}|\ru^{t},\ry^{t})-\KL_{\mathbb{P},\mathbb{G}}(\bar{\rx}_{t}|\ru^{t-1},\ry^{t})\right]\\
    =& \sum_{t=1}^{T}\left[\KL_{\mathbb{P},\mathbb{G}}(\bar{\rx}_{t+1}|\ru^{t},\ry^{t+1}) + \KL_{\mathbb{P},\mathbb{G}}(\ry_{t+1}|\ru^{t},\ry^{t})\right.\\
    &\left.+\KL_{\mathbb{P},\mathbb{G}}(\bar{\rx}_{t}|\ru^{t},\ry^{t+1},\bar{\rx}_{t+1})-\KL_{\mathbb{P},\mathbb{G}}(\bar{\rx}_{t}|\ru^{t-1},\ry^{t})\right]\\
    =&\sum_{t=1}^{T}\left[\KL_{\mathbb{P},\mathbb{G}}(\ry_{t+1}|\ru^{t},\ry^{t}) + \KL_{\mathbb{P},\mathbb{G}}(\bar{\rx}_{t}|\ru^{t},\ry^{t+1},\bar{\rx}_{t+1})\right]\\
    &+\KL_{\mathbb{P},\mathbb{G}}(\bar{\rx}_{T+1}|\ru^{T},\ry^{T+1}) - \KL_{\mathbb{P},\mathbb{G}}(\bar{\rx}_{1}|\ry^{1})\\
    \geq & 0
\end{align*}
where the first equality is by Lemma \ref{lemma:P-minus-G}(a); the second equality is by Lemma \ref{lemma:P-minus-G}(b); the last inequality follows from the nonnegativity of the KL divergence and the fact that $\rx_1$ is Gaussian.
    \end{proof}

    \textit{Proof of inequality \eqref{ineq2}:}
    Now consider another policy $\dQ$ such that:
    \begin{align*}
        \dQ(u_t| u^{t-1},y^{t},\bar{x}_t) = \mathcal{N}(\text{llmsee}(u^{t-1},y^{t},\bar{x}_t),H_t)
    \end{align*}
    where llmsee is the linear least mean-square error estimation of $\ru_t$ given $(\ru^{t-1},\ry^{t},\bar{\rx}_t)$ in the Gaussian measure $\mathbb{G}$.
    We take $\ru_t = \text{llmsee} (\ru^{t-1},\ry^{t},\bar{\rx}_t) + \rh_t$ where $\rh_t$ is a Gaussian random variable with covariance matrix $H_t\succeq 0$ and is independent of $(\rx^{t},\ry^{t},\rz^{t},\ru^{t-1})$. By construction, we have
    \begin{align}\label{eqn:Q-construction-property}
        \dQ(u_t|u^{t-1},y^{t},\bar{x}_t)\dG(u^{t-1},y^{t},\bar{x}_t) = \dG(u^{t},y^{t},\bar{x}_t)
    \end{align}
    Recursively we can define $\mathbb{Q}(u^{T},y^{T+1},\bar{x}^{T+1})$, the measure generated by $\dQ(u_t| u^{t-1},y^{t},\bar{x}_t)$, \eqref{eqn:observable-process-LQG}, \eqref{eqn:enc-observation} and \eqref{eqn:side-info} as:
    \begin{align}
        \dQ(u^{t},y^{t+1},\bar{x}^{t+1}) &= \dP(\bar{x}_{t+1},y_{t+1}|u_{t},\bar{x}_{t})\dQ(u^{t},y^{t},\bar{x}^{t})\label{eqn:Q-construct1} \\
        \dQ(u^{t},y^{t},\bar{x}^{t}) &= \dQ(u_t|u^{t-1},y^{t},\bar{x}_t) \dQ(u^{t-1},y^{t},\bar{x}^{t})\label{eqn:Q-construct2}
    \end{align}
    Next we show the following identity between $\mathbb{G}$ and $\mathbb{Q}$.
    \begin{lemma}\label{lemma:G-Q-identity}
        $\mathbb{Q}(\ru^{t-1},\ry^{t},\bar{\rx}_t) = \mathbb{G}(\ru^{t-1},\ry^{t},\bar{\rx}_t)$
    \end{lemma}
    \begin{proof}
        We prove this by induction. Note that for $t=1$, we have $\mathbb{Q}(\ry_1 ,\bar{\rx}_1) = \mathbb{G}(\ry_1, \bar{\rx}_1)$ since $\bar{\rx}_1 = \mathbb{E}[\rx_1|\ry_1, \rz_1, \ru_1] =\mathbb{E}[\rx_1|\ry_1, \rz_1] $ and $(\rx_1, \ry_1,\rz_1)$ are jointly Gaussian.

Assuming the identity in the lemma, namely, $\mathbb{Q}(\ru^{t-1},\ry^{t},\bar{\rx}_t) = \mathbb{G}(\ru^{t-1},\ry^{t},\bar{\rx}_t)$, holds at time step $t$, we show it also holds for $t+1$.
\begin{align*}
    &\mathbb{Q}(\ru^{t},\ry^{t+1},\bar{\rx}_{t+1})\\
    =&\int_{\bar{x}_{t}} \dQ(u^{t},y^{t+1},\bar{x}_{t+1},\bar{x}_t)\\
    \stackrel{(a)}{=}&\int_{\bar{x}_t} \dP(y_{t+1},\bar{x}_{t+1}|u^{t},y^{t},\bar{x}_t)\dQ(u_t|u^{t-1},y^{t},\bar{x}_t)\\
    &\dQ(u^{t-1},y^{t},\bar{x}_t)\\
    \stackrel{(b)}{=}&\int_{\bar{x}_t}\dP(y_{t+1},\bar{x}_{t+1}|u^{t},y^{t},\bar{x}_t)\dG(u^{t},y^{t},\bar{x}_t)\\
    \stackrel{(c)}{=}&\int_{\bar{x}_t}\dG(u^{t},y^{t+1},\bar{x}_{t+1},\bar{x}_t)\\
    =&\mathbb{G}(\ru^{t},\ry^{t+1},\bar{\rx}_{t+1})
\end{align*}
where equality (a) is by \eqref{eqn:Q-construct1} and \eqref{eqn:Q-construct2}; equality (b) is by the induction hypothesis and \eqref{eqn:Q-construction-property}; equality (c) is due to Lemma \ref{lemma:P-equals-G}.
    \end{proof}
    By Lemma \ref{lemma:G-Q-identity} and \eqref{eqn:Q-construction-property}, we have that $\mathbb{Q}(\ru^{t},\ry^{t},\bar{\rx}_t)=\mathbb{G}(\ru^{t},\ry^{t},\bar{\rx}_t)$. This implies that
    \begin{align*}
        \sum_{t=1}^{T}I_{\mathbb{G}}(\bar{\rx}_t;\ru_t|\ru^{t-1},\ry^{t}) =\sum_{t=1}^{T}I_{\mathbb{Q}}(\bar{\rx}_t;\ru_t|\ru^{t-1},\ry^{t})
    \end{align*}
    Note that for $\mathbb{E}[\Vert \rx_{t+1} \Vert^{2}_{Q_t}]$ we have:
    \begin{align*}
        \mathbb{E}[\Vert \rx_{t+1} \Vert^{2}_{Q_t}] 
        =&\mathbb{E}[\Vert \bar{\rx}_{t+1} + \bar{\rw}_{t+1} \Vert^{2}_{Q_t}]\\
        =&\mathbb{E}[\Vert \bar{\rx}_{t+1}  \Vert^{2}_{Q_t}] + \mathbb{E}[\Vert \bar{\rw}_{t+1}  \Vert^{2}_{Q_t}]
    \end{align*}
    where the last equality holds because $\bar{\rw}_{t+1}$ is independent of $(\ry^{t+1},\rz^{t+1},\ru^{t+1})$ and $\bar{\rx}_{t+1} = \mathbb{E}[\rx_{t+1}|\I_{t+1}]$. So if the second moments of $(\bar{\rx}_t, \ru_t)$ are the same between $\mathbb{P}$ and $\mathbb{Q}$, the control costs are the same between $\mathbb{P}$ and $\mathbb{Q}$. Since $\mathbb{Q}(\ru^{t},\ry^{t},\bar{\rx}_t)=\mathbb{G}(\ru^{t},\ry^{t},\bar{\rx}_t)$ and the second moments under $\mathbb{P} $ and $\mathbb{G}$ are the same for $(\bar{\rx}_t,\ru_t)$, $\mathbb{P}$ and $\mathbb{Q}$ have the same second moments for $(\bar{\rx}_t,\ru_t)$. Hence, the policy constructed by $\dQ$ gives a lower bound for the original optimization problem \eqref{LQG-plain-opt-form}. 

    Next, we show that the policies in $\Omega_1$ can recover the policy constructed by $\dQ$. Note that in $\dQ$, we have
    \begin{align*}
        \ru_t = \text{llmsee} (\ru^{t-1},\ry^{t},\bar{\rx}_t) + \rh_t
    \end{align*}
    where $\rh_t\sim \mathcal{N}(0,H_t)$ is an independent Gaussian random variable. Denote $L'_t$ to be the coefficient of $\rz^{t}$ from $\bar{\rx}_t$ in $\text{llmsee} (\ru^{t-1},\ry^{t},\bar{\rx}_t)$ and $\rc_t = L'_t \rz^{t} + \rh_t$. Note that the coefficient of $\rz^{t}$ can be effectively written as $L_t L^{*}_t$ where $\bar{\rx}_t = \mathbb{E}[\rx_t|\ry^{t},\rz^{t},\ru^{t-1}]$ is a linear combination of $(\ry^{t},\rz^{t},\ru^{t-1})$.
    Hence, conditioning on $(\ru^{t-1},\ry^{t})$ is the same as conditioning on $(\rc^{t-1},\ry^{t})$. This implies the following identity:
    \begin{align*}
        \sum_{t=1}^{T}I_{\mathbb{Q}}(\bar{\rx}_t;\ru_t|\ru^{t-1},\ry^{t}) =\sum_{t=1}^{T}I(\bar{\rx}_t;\rc_t|\rc^{t-1},\ry^{t})
    \end{align*}
    Since the $\rc_t$s are linear encoders as in the policy set $\Omega_1$, and the controller $\ru_{t}$ is a linear function of $(\ru^{t-1},\ry^{t},\rc_t)$, the R.H.S. above is fixed given the linear encoders and side information. Thus, the objective is independent of the choice of controllers.
    In particular, from the separation principle, the certainty equivalence controller is optimal given the linear sensors and side information.
    Therefore, we can replace the controllers in $\dQ$ with the certainty equivalence controller since the objective is fixed and the control performance is better.
    This completes the proof of the inequality and gives the equivalent form stated in Theorem \ref{thm:linear-policy-opt}.


    \textit{Proof of inequality \eqref{ineq3}:}
    Note that $\Omega_1$ is a set of feasible policies, and by the structure of the encoder, we have the following for any policy $\omega_{1}\in\Omega_{1}$:
    \begin{align*}
        \sum_{t=1}^{T}I(\bar{\rx}_t;\rc_t|\rc^{t-1},\ry^{t}) 
        \stackrel{(a)}{=} & \sum_{t=1}^{T}I(\bar{\rx}_t,\rz^{t};\rc_t|\rc^{t-1},\ry^{t})\\
        \stackrel{(b)}{\geq} & \sum_{t=1}^{T}I(\rz^{t};\rc_t|\rc^{t-1},\ru^{t-1},\ry^{t}) \\
        \stackrel{(c)}{\geq} & \sum_{t=1}^{T}I(\rz^{t};\ru_t|\ru^{t-1},\ry^{t})
    \end{align*}
    where equality (a) holds because $\rc_t - (\bar{\rx}_t,\rc^{t-1},\ry^{t}) - \rz^{t}$ is Markov, which is true because, given $(\rc^{t-1},\ry^{t})$, $\rc_t$ can be viewed as $L_t \bar{\rx}_t +\rh_t$ as stated in Remark \ref{remark:linear-controller}; inequality (b) follows from the nonnegativity of mutual information and the fact that $\ru_t$ is the certainty equivalence controller; inequality (c) is due to the fact that $\rc_t$ and $\ru_t$ are feasible encoders and controllers, and hence the proof of \cite[Theorem 1]{multiple-sensor-lower-bound} applies. We can further lower bound the objective by considering the set of all feasible policies. This completes the proof.

\subsection{Proof of Theorem \ref{thm:convex-time-varying}}\label{proof:convex-TV}
From Theorem \ref{thm:linear-policy-opt}, it is sufficient to consider the following optimization problem:
\begin{align*}
    \inf_{\omega\in\Omega_1, J\leq d} \sum_{t=1}^{T}I(\Bar{\rx}_t; \rc_t |\rc^{t-1},\ry^{t})
\end{align*}
The policy set $\Omega_1$ consists of linear encoders and the certainty equivalence controller. For the scalar case, it is sufficient to restrict $H_t > 0$. The objective can be written as linear combinations of the log-determinant of covariance matrices:
\begin{align*}
    \sum_{t=1}^{T}I(\Bar{\rx}_t; \rc_t |\rc^{t-1},\ry^{t})
    =\sum_{t=1}^{T}\frac{1}{2}(\log|\bar{P}_{t|t-1}^{\SI}| - \log |\bar{P}_{t|t}|)
\end{align*}
By equation \eqref{eqn:y1update-for-Pbar} and the Woodbury matrix identity, we have the following equality:
\begin{align}
    \bar{P}_{t+1|t}^{\SI} = &\bar{P}_{t+1|t}- (\bar{P}_{t+1|t}+\bar{W}_{t+1})F_{t+1}^{\T}\nonumber\\
    &(V_{t+1} + F_{t+1}(\bar{P}_{t+1|t}+\bar{W}_{t+1})F_{t+1}^{\T})^{-1}\nonumber\\
    &F_{t+1}(\bar{P}_{t+1|t}+\bar{W}_{t+1})\label{eqn:Pbarone-full}
\end{align}

For the scalar case, this is reduced to 
\begin{align}\label{eqn:Pbar-scalar}
    \bar{P}_{t+1|t}^{\SI} = \bar{P}_{t+1|t} - \frac{F_{t+1}^{2}(\bar{P}_{t+1|t}+\bar{W}_{t+1})^{2}}{V_{t+1}+F_{t+1}^{2}(\bar{P}_{t+1|t}+\bar{W}_{t+1})}
\end{align}
Note that $\bar{P}_{t+1|t} = A_{t}\bar{P}_{t|t}A_{t}^{\T} + \bar{N}_t$ where $\bar{N}_t$ is the covariance matrix of $\bar{\rn}_t$ and that $\bar{\rw}_{t+1} = \bar{\rw}_{t+1|t} - \bar{\rn}_t$, which follows from the update equation for $\bar{\rw}_{t+1}$.

Since $\bar{\rw}_{t+1}$ is independent of $\I_{t+1}$ and $\bar{\rn}_t$ is measurable with respect to $\I_{t+1}$, we have that $\bar{\rw}_{t+1}$ is independent of $\bar{\rn}_t$. Hence, the following covariance matrix identity is true:
\begin{align}\label{eqn:Wbar-update-and-Nbar}
    \bar{W}_{t+1|t} =\bar{W}_{t+1} +\bar{N}_t
\end{align}
Then we can rewrite the objective:
\begin{align*}
    & \sum_{t=1}^{T}\frac{1}{2}(\log\bar{P}_{t|t-1}^{\SI} -\log \bar{P}_{t|t})\\
    =& \sum_{t=1}^{T-1}\frac{1}{2}(\log\bar{P}_{t+1|t}^{\SI} - \log \bar{P}_{t|t})+\frac{1}{2}(\log\bar{P}_{1|0}^{\SI} - \log\bar{P}_{T|T})\\
    =&\sum_{t=1}^{T-1}\frac{1}{2}f_t(\bar{P}_{t|t}) +\frac{1}{2}(\log\bar{P}_{1|0}^{\SI} - \log\bar{P}_{T|T})
\end{align*}
where the last equality is by \eqref{eqn:Pbar-scalar}, and $f_t(\bar{P}_{t|t})=\log(\alpha_t  \bar{P}_{t|t}+\beta_t) - \log(\gamma_t \bar{P}_{t|t} + \delta_t)-\log\bar{P}_{t|t}$ with parameters 
\begin{align*}
    \alpha_{t} &= A_{t}^{2}(V_{t+1}-F_{t+1}^{2}\bar{W}_{t+1})\\
    \beta_t &= \bar{N}_t(V_{t+1} - F_{t+1}^{2}\bar{W}_{t+1})-F_{t+1}^{2}\bar{W}_{t+1}^{2}\\
    \gamma_t &= A_{t}^{2}F_{t+1}^{2}\\
    \delta_t &= V_{t+1}+F_{t+1}^{2}(\bar{N}_{t}+\bar{W}_{t+1})
\end{align*}
Next, we show that $f_t(\bar{P}_{t|t})$ is a convex function of $\bar{P}_{t|t}$ for $ \bar{P}_{t|t}>0$ by evaluating the second derivative directly. 
The first derivative of $f_t(\bar{P}_{t|t})$ is
\begin{align*}
    f_t'(\bar{P}_{t|t}) = \frac{\alpha_t}{\alpha_t \bar{P}_{t|t}+\beta_t} - \frac{\gamma_t}{\gamma_t \bar{P}_{t|t}+\delta_t} - \frac{1}{\bar{P}_{t|t}}
\end{align*}
The second derivative is
\begin{align*}
    f_t''(\bar{P}_{t|t}) &= -\frac{\alpha_t^{2}}{(\alpha_t \bar{P}_{t|t}+\beta_t)^2} + \frac{\gamma_t^{2}}{(\gamma_t \bar{P}_{t|t}+\delta_t)^2}+\frac{1}{\bar{P}_{t|t}^2}\\
    &= \frac{2\alpha_t \beta_t \bar{P}_{t|t}+\beta_t^{2}}{\bar{P}_{t|t}^{2}(\alpha_t \bar{P}_{t|t}+\beta_t)^2} + \frac{\gamma_t^{2}}{(\gamma_t \bar{P}_{t|t}+\delta_t)^2}
\end{align*}
Therefore, it is sufficient to show $\alpha_t \geq0$ and $\beta_t\geq0$ to conclude that $ f_t''(\bar{P}_{t|t})\geq0$ for $ \bar{P}_{t|t}>0$.
From the standard Kalman filter update equation, we have
\begin{align}\label{eqn:Wbar-update-scalar}
    \bar{W}_{t+1}^{-1} = \bar{W}_{t+1|t}^{-1} + \frac{F_{t+1}^{2}}{V_{t+1}} + \frac{C_{t+1}^{2}}{N_{t+1}}
\end{align}
which implies $\bar{W}_{t+1}^{-1} \geq \frac{F_{t+1}^{2}}{V_{t+1}}$, hence $\alpha_t \geq 0$.
Then we show $\beta_t \geq0$. By \eqref{eqn:Wbar-update-and-Nbar}, we have the following equalities for $\beta_t$:
\begin{align*}
    \beta_t &=  \bar{N}_t(V_{t+1} - F_{t+1}^{2}\bar{W}_{t+1})-F_{t+1}^{2}\bar{W}_{t+1}^{2}\\
    &= \bar{N}_t V_{t+1} - F_{t+1}^{2}\bar{W}_{t+1}(\bar{W}_{t+1} + \bar{N}_t)\\
    &= V_{t+1}(\bar{W}_{t+1|t} - \bar{W}_{t+1}) - F_{t+1}^{2}\bar{W}_{t+1}\bar{W}_{t+1|t}
\end{align*}
where the last equality is by \eqref{eqn:Wbar-update-and-Nbar}.
By \eqref{eqn:Wbar-update-scalar}, we can write $\bar{W}_{t+1}$ as
\begin{align*}
    \bar{W}_{t+1} = \frac{\bar{W}_{t+1|t}}{1+(\frac{F_{t+1}^{2}}{V_{t+1}}+\frac{C_{t+1}^{2}}{N_{t+1}})\bar{W}_{t+1|t}}
\end{align*}
Hence, we can write $\beta_t$ as
\begin{align*}
    \beta_t &= \frac{\bar{W}_{t+1|t}^{2}}{1+(\frac{F_{t+1}^{2}}{V_{t+1}}+\frac{C_{t+1}^{2}}{N_{t+1}})\bar{W}_{t+1|t}}\left[ F_{t+1}^{2} + \frac{C_{t+1}^{2}V_{t+1}}{N_{t+1}} -F_{t+1}^{2}\right]\\
    &\geq 0
\end{align*}
Therefore, $f_t''(\bar{P}_{t|t})\geq0$ for $ \bar{P}_{t|t}>0$. $f_t(\bar{P}_{t|t})$ is a convex function for $\bar{P}_{t|t}>0$.
Since the objective function is the sum of $f_t(\bar{P}_{t|t})$ and $-\log\bar{P}_{T|T}$, and $\bar{P}_{1|0}^{\SI}$ is a constant, we have that the objective is jointly convex in $\{\bar{P}_{t|t}>0\}_{t=1}^{T}$.

Next we consider the constraint $\bar{P}_{t|t}\preceq\bar{P}_{t|t-1}^{\SI}$. By \eqref{eqn:Pbarone-full} and the Schur complement, this constraint is equivalent to
\begin{align*}
\begin{bmatrix}
    \bar{P}_{t|t-1} - \bar{P}_{t|t} & (\bar{P}_{t|t-1}+\bar{W}_{t})F_{t}^{\T}\\
    F_{t}(\bar{P}_{t|t-1}+\bar{W}_{t}) & F_{t}(\bar{P}_{t|t-1}+\bar{W}_{t})F_{t}^{\T}+V_{t}
\end{bmatrix}\succeq 0
\end{align*}
For the control cost constraint, for the certainty equivalence controller, the total control cost can be written as 
\begin{align*}
    \sum_{t=1}^{T}\left[\Tr(\Theta_{t}P_{t|t})+\Tr(S_t W_t)\right]+\Tr(\Phi_1 P_{1|0})
\end{align*}
Since $P_{t|t} = \bar{P}_{t|t} + \bar{W}_{t}$, this can be written as
\begin{align*}
    \sum_{t=1}^{T}\left[\Tr(\Theta_{t}\bar{P}_{t|t})+\Tr(\Theta_{t}\bar{W}_{t})+\Tr(S_t W_t)\right]+\Tr(\Phi_1 P_{1|0})
\end{align*}

\subsection{Proof of Theorem \ref{thm:convex-TI}}\label{proof:convex-TI}
Firstly, the arguments for Theorem \ref{thm:linear-policy-opt} also hold for the optimization problem \eqref{LQG-plain-opt-form-TI}. So it is sufficient to consider the policy set $\Omega_1$. We can also restrict $H_t>0$ for the scalar case. Secondly, the certainty equivalence controller is time-invariant for the infinite horizon case since $(A,B)$ is stabilizable. Hence, following the proof of Theorem \ref{thm:convex-time-varying}, the optimization problem \eqref{LQG-plain-opt-form-TI} can be written as:
\begin{align}
    \inf_{\{\bar{P}_{t|t}\}_{t=1}^{\infty}} &\limsup_{T\rightarrow\infty}\frac{1}{T} \left[\sum_{t=1}^{T-1} \frac{1}{2}f_{t}(\bar{P}_{t|t}) + \frac{1}{2}\log\frac{\bar{P}_{1|0}^{\SI}}{\bar{P}_{T|T}}\right]\label{convex-TV-forTI}\\
    \text{s.t. }& \limsup_{T\rightarrow\infty}\frac{1}{T}\left[\sum_{t=1}^{T}\Tr(\Theta \bar{P}_{t|t}) + \eta_1 \right]\leq d \label{cons:limsup-TV-control}\\
    & \begin{bmatrix}
       \bar{P}_{t|t-1}  - \bar{P}_{t|t}& (\bar{P}_{t|t-1}+\bar{W}_t)F^{\T} \\
        F(\bar{P}_{t|t-1}+\bar{W}_t) & F_(\bar{P}_{t|t-1}+\bar{W}_t)F^{\T}+V
    \end{bmatrix} \succeq 0 \label{cons:limsup-TV}\\
        & \bar{P}_{t+1|t} = A\bar{P}_{t|t} A^{\T} + \bar{N}_t \\
    &\bar{P}_{t|t} > 0
\end{align}
where $\eta_1 = \sum_{t=1}^{T}[\Tr(S W) + \Tr(\Theta \bar{W}_t)] $.

Denote the average of $\bar{P}_{t|t}$ as $\bar{P}^{'}_{T} \triangleq  \frac{1}{T}\sum_{t=1}^{T}\bar{P}_{t|t}$ and the set $\mathcal{D} $ as
\begin{align}
    D = \{\bar{P}\geq0: &\Tr(\Theta\bar{P})+\Tr(WS)+\Tr(\Theta\bar{W})\leq d\nonumber\\
    &\begin{bmatrix}
        \bar{P}^{+}-\bar{P} & (\bar{P}^{+}+\bar{W})F^{\T}\\
        F(\bar{P}^{+}+\bar{W}) & F(\bar{P}^{+}+\bar{W})F^{\T}+V
    \end{bmatrix}\succeq 0\}
\end{align}
By the control cost constraint \eqref{cons:limsup-TV-control}, we have that 
\begin{align*}
   \limsup_{T\rightarrow\infty}[ \Tr(\Theta \bar{P}^{'}_{T})+\Tr(S W)+\frac{1}{T}\sum_{t=1}^{T} \Tr(\Theta \bar{W}_t)] \leq d \\
\end{align*}
Since $\bar{W}_t$ converges to $\bar{W}$, we have that for any $\epsilon >0$, there exists $T_{\epsilon}$ such that for any $T\geq T_{\epsilon}$ we have that
\begin{align*}
    \Tr(\Theta \bar{P}^{'}_{T})+\Tr(S W)+ \Tr(\Theta \bar{W}) \leq d +\epsilon\\
\end{align*}

By summing the $T+1$ terms of constraint \eqref{cons:limsup-TV} and dividing it by $T$, and letting $\bar{W}^{'}_{T} \triangleq \frac{1}{T}\sum_{t=1}^{T} \bar{W}_{t}$ and $\bar{N}^{'}_{T} \triangleq \frac{1}{T}\sum_{t=1}^{T} \bar{N}_{t}$ we have the following:
\begin{align*}
    &\begin{bmatrix}
        A^{2}\bar{P}^{'}_{T} +\bar{N}^{'}_{T}-\bar{P}^{'}_{T} & F(A^{2}\bar{P}_{T}^{'}+\bar{W}^{'}_T+\bar{N}^{'}_{T}) \\
        F(A^{2}\bar{P}_{T}^{'}+\bar{W}^{'}_{T}+\bar{N}^{'}_{T}) & F^2(A^{2}\bar{P}_{T}^{'}+\bar{W}^{'}_{T}+\bar{N}^{'}_{T}) + V
    \end{bmatrix} \\
    \succeq & \frac{1}{T}\begin{bmatrix}
        \bar{P}_{T+1|T+1} - \bar{P}_{1|0} & -F(\bar{P}_{1|0}+\bar{W}_{T+1}) \\
        -F(\bar{P}_{1|0}+\bar{W}_{T+1}) & -F^{2}(\bar{P}_{1|0}+\bar{W}_{T+1}) - V
    \end{bmatrix}\\
    \succeq & \frac{1}{T}\begin{bmatrix}
         - \bar{P}_{1|0} & -F(\bar{P}_{1|0}+\bar{W}_{T+1}) \\
        -F(\bar{P}_{1|0}+\bar{W}_{T+1}) & -F^{2}(\bar{P}_{1|0}+\bar{W}_{T+1}) - V
    \end{bmatrix}
\end{align*}
Note that $\bar{W}_t$ converges to $\bar{W}$ and $\bar{N}_t$ converges to $\bar{N}$. For any $\epsilon>0$, there exists $T_{\epsilon}$ such that for all $T\geq T_{\epsilon}$, we have
\begin{align*}
     &\begin{bmatrix}
        A^{2}\bar{P}^{'}_{T} +\bar{N}-\bar{P}^{'}_{T} & F(A^{2}\bar{P}_{T}^{'}+\bar{W}+\bar{N}) \\
        F(A^{2}\bar{P}_{T}^{'}+\bar{W}+\bar{N}) & F^2(A^{2}\bar{P}_{T}^{'}+\bar{W}+\bar{N}) + V
    \end{bmatrix} \\
    \succeq & -\epsilon I
\end{align*}

Define the set $\mathcal{D}_{\epsilon}$ as follows:
\begin{align*}
    &\left\{\bar{P}\geq 0: \Tr(\Theta\bar{P})+\Tr(WS)+\Tr(\Theta\bar{W})\leq d + \epsilon \right. \\
    &\left.
    \begin{bmatrix}
        \bar{P}^{+}-\bar{P} & (\bar{P}^{+}+\bar{W})F^{\T}\\
        F(\bar{P}^{+}+\bar{W}) & F(\bar{P}^{+}+\bar{W})F^{\T}+V
    \end{bmatrix}\succeq -\epsilon I\right\}
\end{align*}

The set $\mathcal{D}_{\epsilon}$ is compact by the same argument as in \cite[Lemma 1]{SDP-TI-Equivalent}. Hence, the sequence $\{\bar{P}^{'}_{T}\}$ has a limit point in $\mathcal{D}_{\epsilon}$ for any $\epsilon >0$. So the sequence $\{\bar{P}^{'}_{T}\}$ has a limit point in $\mathcal{D} = \cap_{\epsilon>0} \mathcal{D}_{\epsilon}$. Let $\{T_{i}\}_{i\in\mathbb{N}}$ be a subsequence such that 
\begin{align*}
    \lim_{i\rightarrow \infty}\bar{P}^{'}_{T_{i}} = \bar{P}^{'}_{\infty}\in\mathcal{D}
\end{align*}
The following lemma implies that we only need to consider the time-invariant function $f(\cdot)$ for the objective.
\begin{lemma}\label{lemma:f-converges}
For any sequence $\{\bar{P}_{t|t}\}_{t= 1}^{\infty}$
with $\bar{P}_{t|t}>0$, we have
\begin{align*}
    \limsup_{T\rightarrow \infty}\frac{1}{T}\sum_{t=1}^{T} f_t(\bar{P}_{t|t}) = \limsup_{T\rightarrow \infty}\frac{1}{T}\sum_{t=1}^{T} f(\bar{P}_{t|t})
\end{align*}
\end{lemma}
\begin{proof}
    See Section \ref{proof:f-converges}.
\end{proof}

Then we have the following chain of inequalities:
\begin{align*}
    &\limsup_{T\rightarrow\infty}\frac{1}{T} \left[\sum_{t=1}^{T-1} \frac{1}{2}f_{t}(\bar{P}_{t|t}) + \frac{1}{2}\log\frac{\bar{P}_{1|0}^{\SI}}{\bar{P}_{T|T}}\right]\\
    \geq & \limsup_{i\rightarrow\infty}\frac{1}{T_{i}+1} \left[\sum_{t=1}^{T_{i}} \frac{1}{2}f_{t}(\bar{P}_{t|t}) + \frac{1}{2}\log\frac{\bar{P}_{1|0}^{\SI}}{\bar{P}_{T_{i}+1|T_{i}+1}}\right]\\
    \geq & \limsup_{i\rightarrow\infty}\frac{1}{T_{i}+1} \sum_{t=1}^{T_{i}} \frac{1}{2}f_{t}(\bar{P}_{t|t}) \\
    = & \limsup_{i\rightarrow\infty}\frac{1}{T_{i}+1} \sum_{t=1}^{T_{i}} \frac{1}{2}f(\bar{P}_{t|t})\\
    \geq & \limsup_{i\rightarrow\infty}\frac{1}{2}f(\bar{P}^{'}_{T_{i}})
\end{align*}
where the first inequality is due to the fact that $\{T_{i}+1\}_{i\in\mathbb{N}}$ is a subsequence of $\mathbb{N}$; the second inequality follows from the boundedness of $\bar{P}_{1|0}^{\SI}$ and the argument in \cite[Lemma 3]{SDP-TI-Equivalent}; the first equality is due to Lemma \ref{lemma:f-converges} and the last inequality is due to the fact that $f(\cdot)$ is a convex function so Jensen's inequality applies.
Hence, we have shown that the following single-letter convex optimization problem provides a lower bound for \eqref{convex-TV-forTI}.
\begin{align}
    \inf_{\bar{P}> 0} & \frac{1}{2}\left[\log(\alpha  \bar{P}+\beta) - \log(\gamma \bar{P} + \delta)-\log\bar{P}\right] \label{convex-TI-single-letter}\\
    \text{s.t. }&  \Tr(\Theta \bar{P})+\Tr(S W)+ \Tr(\Theta \bar{W}) \leq d \\
    & \begin{bmatrix}
        \bar{P}^{+}-\bar{P} & (\bar{P}^{+}+\bar{W})F\\
        F(\bar{P}^{+}+\bar{W}) & F^{2}(\bar{P}^{+}+\bar{W})+V
    \end{bmatrix} \succeq 0 \\
        & \bar{P}^{+} = A\bar{P} A^{\T} + \bar{N} 
\end{align}

Next, we need to show that the single-letter optimization form \eqref{convex-TI-single-letter} is lower bounded by \eqref{convex-TV-forTI}. Consider any $\bar{P}$ that is feasible for \eqref{convex-TI-single-letter}. Denote $\bar{P}^{+} = A^{2}\bar{P}+\bar{N}$ and choose $(L, H)$ such that 
\begin{align*}
    \bar{P}^{-1}-\left( [(\bar{P}^{+}+\bar{W})^{-1}+F^{2}V^{-1}]^{-1}-\bar{W}\right)^{-1} = L^{2}H^{-1}
\end{align*}
We need to show that the following recursions of $\bar{P}_{t|t}$ converge to $\bar{P}$.
\begin{subequations}\label{eqn:recursions-Pbar}
    \begin{align}
    (\bar{P}_{t+1|t}^{\SI}+\bar{W}_{t+1})^{-1} &=(\bar{P}_{t+1|t}+\bar{W}_{t+1})^{-1}+F^{2}V^{-1}\\
    \bar{P}_{t+1|t+1}^{-1} &=(\bar{P}_{t+1|t}^{\SI})^{-1} +L^{2}H^{-1}\\
    \bar{P}_{t+1|t} &= A^{2}\bar{P}_{t|t}+\bar{N}_t
\end{align}
\end{subequations}

\begin{lemma}\label{lemma:convergence-recursion}
    The recursions \eqref{eqn:recursions-Pbar} of $\bar{P}_{t|t}$ converge to $\bar{P}$ for any initial value $P_{1|0}>0$.
\end{lemma}
\begin{proof}
    See Section \ref{proof:convergence-recursion}.
\end{proof}
By Lemma \ref{lemma:convergence-recursion}, for fixed $(L,H)$, the recursions \eqref{eqn:recursions-Pbar} of $\bar{P}_{t|t}$ converge to $\bar{P}$. Therefore, $\bar{P}^{\SI}_{t|t-1}$ also converges to 
\begin{align*}
    \bar{P}^{+} - \frac{F^{2}(\bar{P}^{+}+\bar{W})^{2}}{V+F^{2}(\bar{P}^{+}+\bar{W})}
\end{align*}
The objective mutual information term $I(\Bar{\rx}_t; \rc_t |\rc^{t-1},\ry^{t})$ converges to 
\begin{align*}
    &\frac{1}{2}\log(\bar{P}^{+} - \frac{F^{2}(\bar{P}^{+}+\bar{W})^{2}}{V+F^{2}(\bar{P}^{+}+\bar{W})}) - \frac{1}{2}\log \bar{P}\\
  =&  \frac{1}{2}\log(\frac{\alpha \bar{P}+\beta}{\gamma\bar{P}+\delta}) - \frac{1}{2}\log\bar{P}
\end{align*}
By the Ces\`aro sum, we have:
\begin{align*}
    \lim_{T\rightarrow\infty}\frac{1}{T}\sum_{t=1}^{T}I(\Bar{\rx}_t; \rc_t |\rc^{t-1},\ry^{t}) = \frac{1}{2}\log(\frac{\alpha \bar{P}+\beta}{\gamma\bar{P}+\delta}) -\frac{1}{2} \log\bar{P}
\end{align*}
Similarly, since $\bar{P}_{t|t}$ converges to $\bar{P}$ and $\bar{W}_t$ converges to $\bar{W}$, the average of the control cost also converges:
\begin{align*}
    &\lim_{T\rightarrow\infty}\frac{1}{T}\sum_{t=1}^{T}\left[\Tr(\Theta \bar{P}_{t|t}) + \Tr(S W) + \Tr(\Theta \bar{W}_t) \right]\\
    =& \Tr(\Theta \bar{P}) + \Tr(S W) + \Tr(\Theta \bar{W}) 
\end{align*}
This completes the proof.

\subsection{Proof of Corollary \ref{col:explicit-solution-TI}}\label{proof:explicit-sol-TI}
The first derivative of $f(x)$ is negative for $x>0$ by the following calculation:
\begin{align*}
    f'(x) &= \frac{\alpha}{\alpha x+\beta} - \frac{\gamma}{\gamma x+\delta} - \frac{1}{x}\\
    &=-\frac{\beta}{(\alpha x+\beta)x}- \frac{\gamma}{\gamma x+\delta}
\end{align*}
Hence, the objective is minimized when $\bar{P}$ attains its maximum. Next, we examine the control cost constraint and the Semidefinite constraint.

From the control cost constraint, we have that 
\begin{align*}
    \bar{P}\leq \frac{d-SW}{\Theta}-\bar{W}
\end{align*}

For the Semidefinite constraint, it is equivalent to having nonnegative diagonal elements and a nonnegative determinant for the two-by-two matrix, which yields:
\begin{align}
    (A^{2}-1)\bar{P}+\bar{N} &\geq 0 \label{cons:positive-entry1}\\
    (A^{2}\bar{P}+\bar{N}-\bar{P})[F^{2}(A^{2}\bar{P}+\bar{N}+\bar{W})+V]&\nonumber\\
    -F^{2}(A^{2}\bar{P}+\bar{N}+\bar{W})^{2} &\geq 0\label{cons:positive-determinant}
\end{align}
Denote the determinant as $g(\bar{P})$, which can be simplified as
\begin{align*}
    g(\bar{P}) =& -\bar{P}^{2} F^{2}A^{2}+V\bar{N}-F^{2}(\bar{N}+\bar{W})\bar{W} \\
        &+ \bar{P}\left[(A^{2}-1)V-(A^{2}\bar{W}+\bar{N}+\bar{W})F^{2}  \right].
\end{align*}
The second constraint is $g(\bar{P})\geq 0$.
Since $g(0) = \beta >0$, the polynomial $g(\bar{P})$ has a positive solution $\bar{P}^{*}$. Since $F^{2}A^{2}>0$, we have the constraint that $\bar{P}\leq \bar{P}^{*}$.

Then we consider the positive diagonal element constraint \eqref{cons:positive-entry1}.
When $|A| \geq 1$, the constraint $\eqref{cons:positive-entry1}$ is satisfied for any $\bar{P}\geq 0$. Overall, we have that 
\begin{align*}
    \bar{P}\leq \min\{\bar{P}^{*}, \frac{d-SW}{\Theta} -\bar{W}\}.
\end{align*}

When $d \leq \Theta \bar{W } + SW + \Theta \bar{P}^{*}$, we consider $\bar{P} = \frac{d-SW}{\Theta} -\bar{W}$. The function $f(\bar{P})$ has the value
\begin{align*}
    &f(\frac{d-SW}{\Theta} -\bar{W})\\
    =&\log(\alpha +\frac{\beta}{\frac{d-SW}{\Theta} -\bar{W}}) -\log(\gamma (\frac{d-SW}{\Theta} -\bar{W})+\delta).
\end{align*}

The $\log(\alpha +\frac{\beta}{\frac{d-SW}{\Theta} -\bar{W}})-\log V$ yields $\log\left(A^{2}(1-F^{2}V^{-1}\bar{W})+\frac{\bar{N}(1-F^{2}V^{-1}\bar{W})-F^{2}V^{-1}\bar{W}^{2}}{\frac{d-SW}{\Theta}-\bar{W}}\right)$ and for $\log(\gamma (\frac{d-SW}{\Theta} -\bar{W})+\delta)$, it equals
\begin{align*}
    &\log(A^{2}F^{2}(\frac{d-SW}{\Theta} -\bar{W}) + V+F^{2}(\bar{N}+\bar{W}))\\
    =& \log(A^{2}F^{2}(\frac{d-SW}{\Theta})+V +F^{2}W)
\end{align*}
where the equality holds because $\bar{N} +\bar{W} = A^{2}\bar{W}+W$. Therefore, $\log(\gamma (\frac{d-SW}{\Theta} -\bar{W})+\delta) - \log V$ yields $\log\left(A^{2}F^{2}V^{-1}\frac{d-SW}{\Theta}+1+F^{2}V^{-1}W\right)$. Overall, we have the optimal value as in Corollary \ref{col:explicit-solution-TI}.

When $d > \Theta \bar{W } + SW + \Theta \bar{P}^{*}$, we have that $\bar{P} = \bar{P}^{*}$. By evaluating the objective at $\bar{P}^{*}$, it is easy to verify that it is zero.

\subsection{Proof of Lemma \ref{lemma:f-converges}}\label{proof:f-converges}
Note that $\alpha_t\geq0$, $\beta_t >0$, $\alpha_t\rightarrow \alpha$ and $\beta_t\rightarrow\beta>0$ under the given conditions. Moreover, $\alpha = 0$ implies $A =0$ and hence $\alpha_t = 0$ for all $t$.
There exists a sequence $\epsilon_t\rightarrow 0$ such that 
\begin{align*}
    (1-\epsilon_t)\alpha&\leq \alpha_t\leq(1+\epsilon_t)\alpha\\
    (1-\epsilon_t)\beta&\leq \beta_t\leq(1+\epsilon_t)\beta
\end{align*}
Therefore, we have
\begin{align*}
    1-\epsilon_t\leq \frac{\alpha_t \bar{P}_{t|t}+\beta_t}{\alpha \bar{P}_{t|t}+\beta}\leq 1+\epsilon_t
\end{align*}
This implies that 
\begin{align*}
    \lim_{t\rightarrow \infty}|\log(\alpha_t \bar{P}_{t|t}+\beta_t)-\log(\alpha \bar{P}_{t|t}+\beta)| = 0
\end{align*}
 By the Ces\`aro sum, we have that
\begin{align*}
    \lim_{T\rightarrow\infty}\frac{1}{T}\sum_{t=1}^{T}
    \log(\frac{\alpha_t \bar{P}_{t|t}+\beta_t}{\alpha \bar{P}_{t|t}+\beta}) = 0
\end{align*}

Similarly, since $\gamma_t = \gamma\geq0$, $\delta_t\geq0$, $\delta_t$ converges to $\delta$ and $\delta>0$, we can show the following with the same argument. 
\begin{align*}
\lim_{T\rightarrow\infty}\frac{1}{T}\sum_{t=1}^{T}\log(\frac{\gamma_t \bar{P}_{t|t}+\delta_t}{\gamma \bar{P}_{t|t}+\delta}) = 0
\end{align*}

Therefore, this implies that 
 \begin{align*}
     &\limsup_{T\rightarrow\infty}\frac{1}{T}\sum_{t=1}^{T}f_{t}(\bar{P}_{t|t}) \\
     =&\limsup_{T\rightarrow\infty}\frac{1}{T}\sum_{t=1}^{T}\left(\log(\frac{\alpha_t  \bar{P}_{t|t}+\beta_t}{\gamma_t \bar{P}_{t|t} + \delta_t})-\log\bar{P}_{t|t}\right)\\
     =& \limsup_{T\rightarrow\infty}\frac{1}{T}\sum_{t=1}^{T}\left(\log(\frac{\alpha  \bar{P}_{t|t}+\beta}{\gamma \bar{P}_{t|t} + \delta}) - \log\bar{P}_{t|t}\right)\\
     =&\limsup_{T\rightarrow\infty}\frac{1}{T}\sum_{t=1}^{T}f(\bar{P}_{t|t})
 \end{align*}

\subsection{Proof of Lemma \ref{lemma:convergence-recursion}}
\label{proof:convergence-recursion}
We first set up the parameters and recursion functions.
Let $c_1 = \frac{F^{2}}{V}$, $c_2 = \frac{C^{2}}{N}$ and $c_3 = \frac{L^{2}}{H}$. Note that the sequences $\bar{W}_t$ and $\bar{N}_t$ satisfy the following recursions:
\begin{align*}
    \bar{W}_{t+1|t} &= A^{2}\bar{W}_t +W\\
    \bar{W}_{t+1}^{-1} &= \bar{W}_{t+1|t}^{-1} + c_1+c_2\\
    \bar{N}_t &= \bar{W}_{t+1|t} - \bar{W}_{t+1}
\end{align*}
From \cite{kailath2000linear}, we have that the sequence $\bar{W}_t$ converges to $\bar{W}$ and hence $\bar{N}_t$ converges to $\bar{N}$. Since $P_{t|t} = \bar{P}_{t|t} +\bar{W}_t$ and $\bar{W}_{t+1|t} - \bar{W}_{t+1} = \bar{N}_t$, we have
\begin{align*}
    \bar{P}_{t+1|t} +\bar{W}_{t+1} &= A^{2}\bar{P}_{t|t}+\bar{N}_t +\bar{W}_{t+1}\\
    &= A^{2}\bar{P}_{t|t}+ \bar{W}_{t+1|t} \\
    &= A^{2}(\bar{P}_{t|t}+\bar{W}_t)+W\\
    &= A^{2}P_{t|t} +W
\end{align*}
By the recursions \eqref{eqn:recursions-Pbar} of $\bar{P}_{t|t}$, we can write $\bar{P}_{t+1|t+1}$ as a function of $P_{t|t} $ as follows:
\begin{equation}\label{eqn:P-to-Pbar}
    \bar{P}_{t+1|t+1} = \frac{\rho(P_{t|t})-\bar{W}_{t+1}}{1+c_{3}(\rho(P_{t|t})-\bar{W}_{t+1})} 
\end{equation}
where $\rho(x ) = \frac{A^{2} x + W}{1+c_{1}(A^{2}x+W)}$. Equivalently, we can write $P_{t+1|t+1}$ as a function of $P_{t|t}$ as follows:
\begin{align*}
    P_{t+1|t+1 } &= \phi_t (\rho(P_{t|t}))\\
    &\triangleq \frac{\rho(P_{t|t})-\bar{W}_{t+1}}{1+c_{3}(\rho(P_{t|t})-\bar{W}_{t+1})}  +\bar{W}_{t+1}
\end{align*}
Next, we show the convergence of recursions \eqref{eqn:recursions-Pbar} for the $(L,H)$ constructed from any feasible $\bar{P}$. We first show that $\bar{P}_{t|t}$ is bounded from below by a positive constant for sufficiently large $t$, which is used to ensure the limit inferior of $P_{t|t}$ is a valid input for the function $\phi_t (\rho(P_{t|t}))$. Since the function $\rho(x)$ can be written as:
\begin{align*}
    \rho(x ) &= \frac{A^{2} x + W}{1+c_{1}(A^{2}x+W)}\\
    &= \frac{1}{c_1}-\frac{1}{c_1+c_1^{2}(A^{2}x+W)}
\end{align*}
we have that $\rho(x)$ is an increasing function for $x>0$. Since $P_{t|t} > \bar{W}_t$, the following inequalities are true:
\begin{align*}
    \rho(P_{t|t}) \geq &\rho(\bar{W}_t)\\
            =& \frac{A^{2} \bar{W}_t + W}{1+c_{1}(A^{2}\bar{W}_t+W)}\\
            > & \frac{A^{2} \bar{W}_t + W}{1+(c_{1}+c_{2})(A^{2}\bar{W}_t+W)}\\
            =&\bar{W}_{t+1}
\end{align*}
Since $\bar{W}_t$ converges to $\bar{W}$, for sufficiently large $T$, there exists a positive constant $p$ such that $\rho(P_{t|t}) - \bar{W}_{t+1}\geq \rho(\bar{W}_t)-\bar{W}_{t+1}>p$ for all $t\geq T$. Combining \eqref{eqn:P-to-Pbar}, this implies that $\bar{P}_{t|t}$ is bounded from below by a positive constant for sufficiently large $t$, which further implies that 
\begin{align*}
    \liminf_{t\rightarrow\infty} P_{t|t} = \liminf_{t\rightarrow\infty} (\bar{P}_{t|t} + \bar{W}_{t})
    > \bar{W}
\end{align*}
The limit inferior of $P_{t|t}$ is strictly larger than $\bar{W}$, the limit of $\bar{W}_{t}$.

Next, we show that $P_{t|t}$ is a bounded sequence. This is true since $P_{t+1|t+1} \leq \rho(P_{t|t}) -\bar{W}_{t+1} + \bar{W}_{t+1} < \frac{1}{c_1}$. Hence, $P_{t|t}$ is a bounded sequence between 0 and $\frac{1}{c_1}$ for large $t$. Denote the limit inferior and limit superior of $P_{t|t}$ as follows
\begin{align*}
     a = \liminf_{t\rightarrow\infty}P_{t|t} \quad b = \limsup_{t\rightarrow\infty}P_{t|t}
\end{align*}
By what we have argued, we have the following inequalities:
\begin{align*}
    \bar{W} < a \leq b \leq \frac{1}{c_1}
\end{align*}

The following step is to show that both $a$ and $b$ are fixed points of the function $\phi_{\infty}(\rho(x))$ with $\bar{W}_{\infty} = \bar{W}$ and $x\geq \bar{W}$.
Since $\bar{W}_t$ converges to $\bar{W}$ and $a$ is the limit inferior of $P_{t|t}$, for any $\epsilon >0$, there exists a sufficiently large $T_{\epsilon}$ such that for any $t\geq T_{\epsilon}$, we have that $|\bar{W} - \bar{W}_t|<\epsilon$ and $P_{t|t} > a-\epsilon$. Pick $\epsilon$ small enough such that $a-\epsilon > \bar{W}+\epsilon$. 
Consider the following compact intervals:
\begin{align*}
    \mathcal{K}\times\mathcal{J} = [a-\epsilon,\frac{1}{c_1}] \times [\bar{W}-\epsilon, \bar{W}+\epsilon]
\end{align*}
Denote the function $\phi(P_{t|t},\bar{W}_{t+1}) \triangleq \phi_t(\rho(P_{t|t}))$, i.e. $\phi(\cdot)$ is the same as $\phi_t(\rho(\cdot))$ for fixed $\bar{W}_{t+1}$. One can check that for any $(x,y)\in \mathcal{K}\times\mathcal{J}$, we have $\phi(x,y)>0$. So $\phi$ is a continuous and positive function on this interval. This implies that $\phi_{t}(\rho(P_{t|t}))$ converges uniformly to $\phi_{\infty}(\rho(P_{t|t}))$ for any $P_{t|t}$ in $\mathcal{K}$. In particular, we have the following inequalities for sufficiently large $t$:
\begin{align*}
    a-\epsilon \leq \phi_t(\rho( P_{t|t})) \leq  b+\epsilon
\end{align*}
Consider a subsequence $\{P_{t_{i}|t_{i}}\}_{i\in\mathbb{N}}$ of $P_{t|t}$ such that $P_{t_i | t_i}$ converges to $a$. The inequality above implies:
\begin{align*}
    a-\epsilon \leq \phi_{\infty}(\rho(a))
\end{align*}
Similarly, we can show:
\begin{align*}
    \phi_{\infty}(\rho(b)) \leq b+\epsilon
\end{align*}
Since $\phi_t(\rho(x))$ is an increasing function for $x\geq \bar{W}_t$ and $a-\epsilon\leq P_{t|t}\leq b+\epsilon$, for sufficiently large $t$, we have: 
\begin{align*}
    \phi_t(\rho(a-\epsilon))\leq \phi_t(\rho(P_{t|t}))= P_{t+1|t+1} \leq \phi_t(\rho(b+\epsilon))
\end{align*}
Taking the limit inferior and the limit superior of $P_{t+1|t+1}$ gives the following inequality:
\begin{align*}
    \phi_{\infty}(\rho(a-\epsilon))\leq a\leq b \leq \phi_{\infty}(\rho(b+\epsilon))
\end{align*}
Since this holds for arbitrarily small $\epsilon$, taking $\epsilon\rightarrow 0$ yields:
\begin{align*}
    a &= \phi_{\infty}(\rho(a))\\
    b &= \phi_{\infty}(\rho(b))
\end{align*}
That is, both $a$ and $b$ are fixed points for the function $\phi_{\infty}(\rho(x))$ with $x \in [a,\frac{1}{c_1}]$. 

The last step is to show that there is only one fixed point for $\phi_{\infty}(\rho(x))$ with $x \geq \bar{W}$. Firstly, $\phi_{\infty}(\rho(x))$ is a composition of two concave increasing functions $\rho(x)$ and $\phi_{\infty}(x)$ both for $x >\bar{W}$. Since $\rho(\bar{W})>\bar{W}$, $\phi_{\infty}(\rho(x))$ is also an increasing concave function for $x\geq\bar{W}$. Recall that $\rho(\bar{W}) > \bar{W}$, so we have $\phi_{\infty}(\rho(\bar{W}))-\bar{W} >0$. Also recall that $\phi_{\infty}(\rho(x)) < \frac{1}{c_1}$ for any $x\geq \bar{W}$. Hence, we have that $\phi_{\infty}(\rho(x)) - x\leq \frac{1}{c_1} -x < 0$ for $x>\frac{1}{c_1}$. Combining the concavity of $\phi_{\infty}(\rho(x))$, we argue that there is only one fixed point for $\phi_{\infty}(\rho(x))$ with $x\geq \bar{W}$. Since $\bar{W}<a\leq b$, this implies that $a=b$. Hence, the sequence $P_{t|t}$ converges to the fixed point of $\phi_{\infty}(\rho(x))$ with $x\geq \bar{W}$. Recall that $P_{t|t} = \bar{P}_{t|t} + \bar{W}_t$ and $\bar{W}_t$ converges t o $\bar{W}$. The sequence $\bar{P}_{t|t}$ also converges. Note that the arguments are independent of the initial value $P_{1|0} $, and hence they hold for any $P_{1|0} >0 $.


\section{Conclusion}
In this paper, we study the minimum rate required for a linear system with partial observation and side information to achieve a certain distortion requirement. We consider the LQG plant with control cost and the Gaussian-Markov source with the weighted mean-square error. We characterize the sufficient policy set for optimizing the causally conditioned directed information lower bound for the rate. We show that the sufficient policy set consists of linear encoders, which are effectively linear functions of the estimate of the plant state with all the observations and independent Gaussian noises, and the certainty equivalence controller for the LQG plant and the conditional mean decoder for the Gaussian-Markov source. The sufficient policy result generalizes the past results regarding the single encoder and single controller/decoder settings. We further show that for the scalar case, the resulting optimization problems admit convex formulations for both finite-horizon time-varying and infinite-horizon time-invariant systems. For the time-invariant system, we show that the optimization problem admits a single-letter form. Our result helps to further understand the sequential rate-distortion tradeoffs for the networked case, where there are multiple encoders and a single decoder.
A direct and interesting future direction is to characterize the optimal policy set of the networked case. One can start with the two-encoder setting with a full-observation sensor and a partial-observation sensor. 
\section*{ACKNOWLEDGMENTS}
The author used GPT-6 for language editing. This work was partly supported by ARO Award W911NF2310062, ONR Award N000142412542, NSF CAREER Award 2443857, and the 6G$@$UT center within the Wireless Networking and Communications Group (WNCG) at the University of Texas at Austin.

\bibliographystyle{ieeetr}
\bibliography{IEEEabrv,Refs.bib}

@ARTICLE{lqgsdp,
  author={Tanaka, Takashi and Esfahani, Peyman Mohajerin and Mitter, Sanjoy K.},
  journal={IEEE Transactions on Automatic Control}, 
  title={LQG Control With Minimum Directed Information: Semidefinite Programming Approach}, 
  year={2018},
  volume={63},
  number={1},
  pages={37-52},
  doi={10.1109/TAC.2017.2709618}}

@ARTICLE{partial-observation-GM,
  author={Stavrou, Photios A. and Skoglund, Mikael},
  journal={IEEE Transactions on Automatic Control}, 
  title={Indirect NRDF for Partially Observable Gauss–Markov Processes With MSE Distortion: Characterizations and Optimal Solutions}, 
  year={2024},
  volume={69},
  number={9},
  pages={5867-5882},
  doi={10.1109/TAC.2024.3364028}}

@ARTICLE{GMsdp,
  author={Tanaka, Takashi and Kim, Kwang-Ki K. and Parrilo, Pablo A. and Mitter, Sanjoy K.},
  journal={IEEE Transactions on Automatic Control}, 
  title={Semidefinite Programming Approach to Gaussian Sequential Rate-Distortion Trade-Offs}, 
  year={2017},
  volume={62},
  number={4},
  pages={1896-1910},
  doi={10.1109/TAC.2016.2601148}}

@ARTICLE{LQG-sideinfo1,
  author={Sabag, Oron and Tian, Peida and Kostina, Victoria and Hassibi, Babak},
  journal={IEEE Transactions on Automatic Control}, 
  title={Reducing the LQG Cost With Minimal Communication}, 
  year={2023},
  volume={68},
  number={9},
  pages={5258-5270},
  doi={10.1109/TAC.2022.3220511}}

@INPROCEEDINGS{LQG-sideinfo2,
  author={Cuvelier, Travis C. and Tanaka, Takashi},
  booktitle={2021 55th Annual Conference on Information Sciences and Systems (CISS)}, 
  title={Rate of Prefix-free Codes in LQG Control Systems with Side Information}, 
  year={2021},
  volume={},
  number={},
  pages={1-6},
  doi={10.1109/CISS50987.2021.9400217}}

@INPROCEEDINGS{multiple-sensor-lower-bound,
  author={Li, Sijie and Tanaka, Takashi and Kim, Hyeji},
  booktitle={2025 IEEE International Symposium on Information Theory (ISIT)}, 
  title={Lower Bound of Networked Linear Quadratic Gaussian Plant with Two Linear Sensors and One Controller}, 
  year={2025},
  volume={},
  number={},
  pages={1-6},
  doi={10.1109/ISIT63088.2025.11195408}}

@INPROCEEDINGS{SDP-TI-Equivalent,
  author={Tanaka, Takashi},
  booktitle={2015 IEEE Conference on Control Applications (CCA)}, 
  title={Semidefinite representation of sequential rate-distortion function for stationary Gauss-Markov processes}, 
  year={2015},
  volume={},
  number={},
  pages={1217-1222},
  doi={10.1109/CCA.2015.7320778}}

@book{kailath2000linear,
  title={Linear Estimation},
  author={Kailath, T. and Sayed, A.H. and Hassibi, B.},
  isbn={9780130224644},
  lccn={99047033},
  series={Prentice-Hall information and system sciences series},
  url={https://books.google.com/books?id=zNJFAQAAIAAJ},
  year={2000},
  publisher={Prentice Hall}
}

@ARTICLE{lqg-epi,
  author={Kostina, Victoria and Hassibi, Babak},
  journal={IEEE Transactions on Automatic Control}, 
  title={Rate-Cost Tradeoffs in Control}, 
  year={2019},
  volume={64},
  number={11},
  pages={4525-4540},
  doi={10.1109/TAC.2019.2912256}}

@ARTICLE{lqg-scalar,
  author={Tatikonda, S. and Sahai, A. and Mitter, S.},
  journal={IEEE Transactions on Automatic Control}, 
  title={Stochastic linear control over a communication channel}, 
  year={2004},
  volume={49},
  number={9},
  pages={1549-1561},
  doi={10.1109/TAC.2004.834430}}

@article{lqgceo,
author = {Jung, Hyunho and Pedram, Ali Reza and Cuvelier, Travis C. and Tanaka, Takashi},
title = {Optimized data rate allocation for dynamic sensor fusion over resource constrained communication networks},
journal = {International Journal of Robust and Nonlinear Control},
volume = {33},
number = {1},
pages = {237-263},
doi = {https://doi.org/10.1002/rnc.6076},
url = {https://onlinelibrary.wiley.com/doi/abs/10.1002/rnc.6076},
eprint = {https://onlinelibrary.wiley.com/doi/pdf/10.1002/rnc.6076},
year = {2023}
}

@ARTICLE{1.254-gap,
  author={Silva, Eduardo I. and Derpich, Milan S. and Østergaard, Jan and Encina, Marco A.},
  journal={IEEE Transactions on Automatic Control}, 
  title={A Characterization of the Minimal Average Data Rate That Guarantees a Given Closed-Loop Performance Level}, 
  year={2016},
  volume={61},
  number={8},
  pages={2171-2186},
  doi={10.1109/TAC.2015.2500658}}

@INPROCEEDINGS{vector-achievable-scheme,
  author={Tanaka, Takashi and Johansson, Karl Henrik and Oechtering, Tobias and Sandberg, Henrik and Skoglund, Mikael},
  booktitle={2016 IEEE International Symposium on Information Theory (ISIT)}, 
  title={Rate of prefix-free codes in LQG control systems}, 
  year={2016},
  volume={},
  number={},
  pages={2399-2403},
  doi={10.1109/ISIT.2016.7541729}}

@misc{longversion,
    title={Rate Lower Bounds for Multi-Sensor Networked
           Control and Their Application to Decoder Side Information},
    author={Sijie Li and Takashi Tanaka and Hyeji Kim},
    year={2026},
    note   = {arXiv:2607.04172},
    url={https://arxiv.org/abs/2607.04172}
}

@ARTICLE{Kramer2003,
  author={Kramer, G.},
  journal={IEEE Transactions on Information Theory}, 
  title={Capacity results for the discrete memoryless network}, 
  year={2003},
  volume={49},
  number={1},
  pages={4-21},
  doi={10.1109/TIT.2002.806135}}

@book{cover2012elements,
  title={Elements of Information Theory},
  author={Cover, T.M. and Thomas, J.A.},
  isbn={9781118585771},
  lccn={2005047799},
  url={https://books.google.com/books?id=VWq5GG6ycxMC},
  year={2012},
  publisher={Wiley}
}

@ARTICLE{DO2012,
  author={Derpich, Milan S. and {\O}stergaard, Jan},
  journal={IEEE Transactions on Information Theory}, 
  title={Improved Upper Bounds to the Causal Quadratic Rate-Distortion Function for {Gaussian} Stationary Sources}, 
  year={2012},
  volume={58},
  number={5},
  pages={3131-3152},
  doi={10.1109/TIT.2012.2184669}}

@ARTICLE{CCKS2022,
  author={Charalambous, Charalambos D. and Charalambous, Themistoklis and Kourtellaris, Christos and van Schuppen, Jan H.},
  journal={IEEE Transactions on Information Theory}, 
  title={Complete Characterization of {Gorbunov} and {Pinsker} Nonanticipatory Epsilon Entropy of Multivariate {Gaussian} Sources: Structural Properties}, 
  year={2022},
  volume={68},
  number={3},
  pages={1440-1464},
  doi={10.1109/TIT.2021.3124888}}

@article{GP1973,
  author  = {Gorbunov, A. K. and Pinsker, M. S.},
  title   = {Nonanticipatory and Prognostic Epsilon Entropies
             and Message Generation Rates},
  journal = {Problems of Information Transmission},
  volume  = {9},
  number  = {3},
  pages   = {184--191},
  year    = {1973}
}

@article{GP1974,
  author  = {Gorbunov, A. K. and Pinsker, M. S.},
  title   = {Prognostic Epsilon Entropy of a {Gaussian} Message
             and a {Gaussian} Source},
  journal = {Problems of Information Transmission},
  volume  = {10},
  number  = {2},
  pages   = {93--109},
  year    = {1974}
}

@INPROCEEDINGS{LK2020,
  author={Lev, Omri and Khina, Anatoly},
  booktitle={2020 International Symposium on Information Theory and Its Applications (ISITA)}, 
  title={{Gauss–Markov} Source Tracking with Side Information: Lower Bounds}, 
  year={2020},
  volume={},
  number={},
  pages={66-70},
  doi={}}

@misc{lqgsdp-arxiv,
      title={LQG Control with Minimum Directed Information: Semidefinite Programming Approach}, 
      author={Takashi Tanaka and Peyman Mohajerin Esfahani and Sanjoy K. Mitter},
      year={2017},
      eprint={1510.04214},
      archivePrefix={arXiv},
      primaryClass={math.OC},
      note={arXiv:1510.04214v3},
      url={https://arxiv.org/abs/1510.04214}, 
}

\begin{IEEEbiographynophoto}{Sijie Li} (Graduate Student Member, IEEE) received the B.S. degree in Mathematics and Information Engineering from the Chinese University of Hong Kong, Hong Kong, in 2022. He is currently a PhD student at the Electrical and Computer Engineering Department of the University of Texas at Austin, affiliated with the Wireless Networking and Communications Group.
\end{IEEEbiographynophoto}

\begin{IEEEbiography}[{\includegraphics[width=1in,height=1.25in,clip,keepaspectratio]{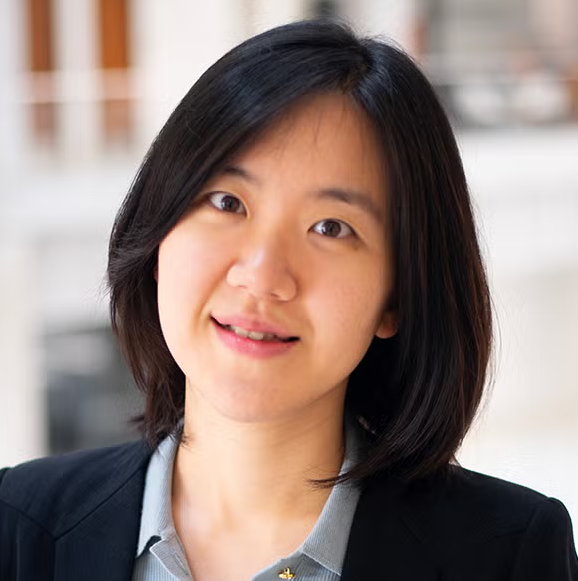}}]{Hyeji Kim} (Senior Member, IEEE) is an Assistant Professor in the Chandra Department of Electrical and Computer Engineering at the University of Texas at Austin. She received her Ph.D. degree in Electrical Engineering from Stanford University in 2016. Before joining UT Austin in 2020, she was a Postdoctoral Researcher at the University of Illinois at Urbana-
Champaign (2016–2018) and a Researcher at Samsung AI Research Cambridge (2018–2020). At UT Austin, she is a member of the Wireless Networking and Communications Group (WNCG) and an affiliate of the Machine Learning Laboratory (MLL). Her research lies at the intersection of information theory, machine learning, and wireless communications. She received the NSF CAREER Award in 2024.
\end{IEEEbiography}

\end{document}